\documentclass[sigconf]{acmart}

\usepackage[english]{babel}
\renewcommand\footnotetextcopyrightpermission[1]{} 
\acmDOI{}
\acmISBN{}
\acmPrice{}
\usepackage[english]{babel}

\usepackage{times}  
\usepackage{hyperref}
\usepackage{color}
\usepackage{graphicx} 
\usepackage{subcaption}
\usepackage{epsfig}
\usepackage[inline]{enumitem}
\usepackage{xspace}
\usepackage{amsmath}
\usepackage{algorithm}
\usepackage{algorithmic}
\usepackage{outlines}
\usepackage[subtle]{savetrees}
\usepackage[normalem]{ulem}
\usepackage{tikz}

\hypersetup{pdfstartview=FitH,pdfpagelayout=SinglePage}

\usepackage{xspace}
\newcommand{\sys}{DBO\xspace}

\newcommand{\pg}[1] {{\textcolor{red}{PG: #1}}}

\newcommand{\attn}[1] {{\textcolor{black}{ #1}}}

\newcommand{\cut}[1]{}

\def\compactify{\itemsep=0pt \topsep=0pt \partopsep=0pt \parsep=0pt}
\let\latexusecounter=\usecounter

\newtheorem{theorem}{Theorem}
\newtheorem{definition}{Definition}

\newtheorem{corollary}{Corollary}
\usepackage[small,compact]{titlesec}
\titlespacing*{\section}{0pt}{4pt}{4pt}
\titlespacing*{\subsection}{0pt}{3pt}{3pt}
\titlespacing*{\subsubsection}{0pt}{3pt}{3pt}

\begin{document}




\title{\LARGE DBO: Response Time Fairness for Cloud-Hosted Financial Exchanges\vspace{-2mm}}


\author{Prateesh Goyal}
\affiliation{Microsoft}

\author{Eashan Gupta}
\affiliation{Microsoft, UIUC}

\author[]{Ilias Marinos}
\affiliation{Microsoft}

\author[]{Chenxingyu Zhao}
\affiliation{University of Washington}

\author[]{Radhika Mittal}
\affiliation{UIUC}

\author[]{Ranveer Chandra}
\affiliation{Microsoft}

\date{\vspace{5mm}}

\begin{abstract}
In this paper, we consider the problem of hosting financial exchanges in the cloud. 
Financial exchanges require predictable, equal latency to all market participants to ensure fairness for various tasks, such as high speed trading.
However, it is extremely difficult to ensure equal latency to all market participants in existing cloud deployments, because of various reasons, such as congestion, and unequal network paths. 
In this paper, we address the unfairness that stems from lack of determinism in cloud networks.
We argue that predictable or bounded latency is not necessary to achieve  fairness. Inspired by the use of logical clocks in distributed systems, we present Delivery Based Ordering (DBO), a new approach that ensures fairness by instead correcting for differences in latency to the participants. We evaluate DBO both in our hardware test bed and in a public cloud deployment and demonstrate that it is feasible to achieve guaranteed fairness and sub-100\textmu s latency while operating at high transaction rates. 

\end{abstract}
\maketitle
\section{Introduction}

Major financial exchanges such as NASDAQ, Chicago Mercantile Exchange (CME), and London Stock Exchange (LSE) have recently expressed interest in migrating their workloads to the cloud aiming to significantly reduce their capital expenditure, improve scalability and reduce operational burden. Major market participants of such  exchanges would also benefit from such migration as they are also maintaining an expensive on-premise infrastructure for data analysis, and regression modelling to formulate their trading strategies. For cloud providers such as Amazon, Google, and Microsoft, this is a big business opportunity. Migrating financial exchanges to the cloud is a mutually beneficial undertaking for all parties involved. 
%
%

 To this end, cloud providers and financial exchanges have announced long-term partnerships to facilitate such a move~\cite{nasdaq_cme_an, nasdaq_aws}. Both parties perceive that this migration will be quite challenging, especially when considering all different workloads (businesses) that are currently accommodated in the exchanges' on-premise infrastructure. 
 In this paper, we focus on ``speed race'' ~\cite{frequent_batch_auctions, libra} trading which is an important and highly profitable business for both the financial exchanges and the market traders. Briefly, `speed race' trading is a form of systematic electronic trading where market participants (``MPs'') use high-performance computers to execute strategies that aim to rapidly react and exploit new opportunities presented in the market (e.g., due to volatility, price discrepancies etc). Speed race traders, also known as High-Frequency Traders invest large amounts of money for hardware, systems and algorithmic development to achieve impressively low reaction times (\textmu s- or even ns- scale). This trading business is only viable if market participants can compete in a \textit{fair} playground guaranteed by the Central Exchange Server (CES) operators. Equality of opportunity -- fairness -- in such case means that all market participants must get provably simultaneous access to market data, as well as their subsequent trades must be executed in the exact order they were generated (i.e. placed in the wire). 

 With on-premise deployments financial exchanges guarantee fairness for speed race trading by guaranteeing equal bi-directional latency to the relevant market participants.  Exchanges go to a great extent to ensure fairness for their co-located MP customers; it is not uncommon, for example, to use layer-1 fan-out switches for market data stream replication and equal-length cables to all co-located MPs. On the contrary, public cloud datacenter networks do not provide such guarantees as they were originally designed for a heterogeneous, multi-tenant environment, aiming to accommodate diverse workloads. Even if the MPs are located within the same cloud region as the CES, it is hard to guarantee that the latency between CES and various MPs will be the same. Copper and fiber optics cables are not necessarily of equal length, network traffic is not evenly balanced among the different paths, multiple vendors' network elements have different performance characteristics, network oversubscription is still common, and network quality of service mechanisms for concurrent workloads are only best effort.

This problem has recently received significant attention from the academic community. Proposed solutions aim to achieve fairness by attempting to provide equal (yet inflated) bi-directional latencies in the cloud relying on tight clock synchronization and buffering for market data delivery (~\cite{cloudex}). As we explain later, such approaches are fragile because latencies in datacenter networks are not only variable, but also unbounded. Other proposals, require  intrusive modifications to existing CES implementations to work. 

In this paper, we seek to address the problem of fairness for speed race trading in cloud environments. Our key insight is that equal bi-directional latencies are not strictly required to achieve fairness. 
For speed trading, instead of ex-ante equalizing latency, we can post facto correct for any latency differences in delivery of data by ordering trades differently. We introduce logical \emph{delivery clocks} that track time at each trader relative to when market data were received. 
We present \textit{Delivery Based Ordering}, a system that uses delivery clocks to order trades and achieve guaranteed fairness in network topologies where latency is non-deterministic and unbounded. 

We implement a real DBO system, which we evaluate on a bare-metal server testbed leveraging programmable NICs. We also evaluate DBO in a public cloud deployment using standard VMs: our system achieves guaranteed fairness and sub-100us p999 latency while servicing 125K trades per second.

\section{Background}



We begin with discussing the challenges in hosting financial exchanges on the cloud, that are derived from our discussions with three major financial exchanges (all are among the top 10 exchanges in the world by trading volume) and our review of papers from financial academic community~\cite{frequent_batch_auctions, libra,burdisch_working, fragile} as well as industry papers~\cite{iex_cost_report, signals_threads}. 

\noindent\textbf{Why is moving to the cloud so hard?} \textit{Short Answer: It is hard to achieve fairness in cloud.} A key customer/business for any major financial exchange is High Frequency Traders (HFTs). At a high level, high frequency traders aim to process incoming market data feed from the exchange server and place trade orders as fast as possible. These traders are engaged in what is known as \emph{speed races} where they are competing for the same trading opportunity, trying to get their trade orders ahead of competition. There is an arms race in high frequency traders to respond to market data the fastest~\cite{frequent_batch_auctions}. HFTs are becoming faster with time,
even minor differences in latency (microsecond level) for market data delivery and trade orders can give a trader significant advantage/disadvantage over the others ~\cite{burdisch_working, fragile, signals_threads}.
Allowing such traders to trade fairly is critical for any exchange to attract HFTs that bring significant liquidity to the exchange. However, cloud environments exhibit variable network latency. This can be due to several reasons, such as congestion in the network, networks paths with unequal hops, etc. Ensuring such fairness in cloud environments is thus very challenging. Exchanges not only want fairness, to speed up price discovery; they also want low latency. The latency requirements depend on the exchange, from sub-100 microsecond to millisecond ~\cite{iex_cost_report}.

\noindent\textbf{How exchanges enable fair speed racing today?} \textit{Short Answer: Equal bi-directional latency.} Major exchanges operate their own datacenters. HFT traders that want to engage in speed trading colocate\footnote{Exchanges support colocation for a limited number of participants. The exact numbers are confidential, but the number is in 10s to less than couple of hundred depending on the exchange.} to the exchange datacenter.

The central exchange server (CES) produces a real time market data feed and distributes it to all the colocated participants (MPs). The exchange datacenter is optimized to ensure that participants get all the market data points at the same time. Further, exchanges ensure that all the trades placed by the participants experience the same latency to the exchange server. The exchange server simply processes the trades in a first-come-first-serve (FCFS) manner. Optimizing datacenters to provide such equal bi-directional latency is expensive~\cite{iex_cost_report}. 
As a result of high cost, exchanges charge a huge premium for such colocation (NASDAQ charges \$600,000 per customer for colocation and direct data feed~\cite{iex_cost_report}). 
These high premiums create a barrier for entry into the high frequency trading world. Major exchanges are also interested in opening up regional exchanges but the cost of creating a new regional datacenter is prohibitively high.

\begin{figure}[t]
\centering
    \includegraphics[width=0.6\columnwidth]{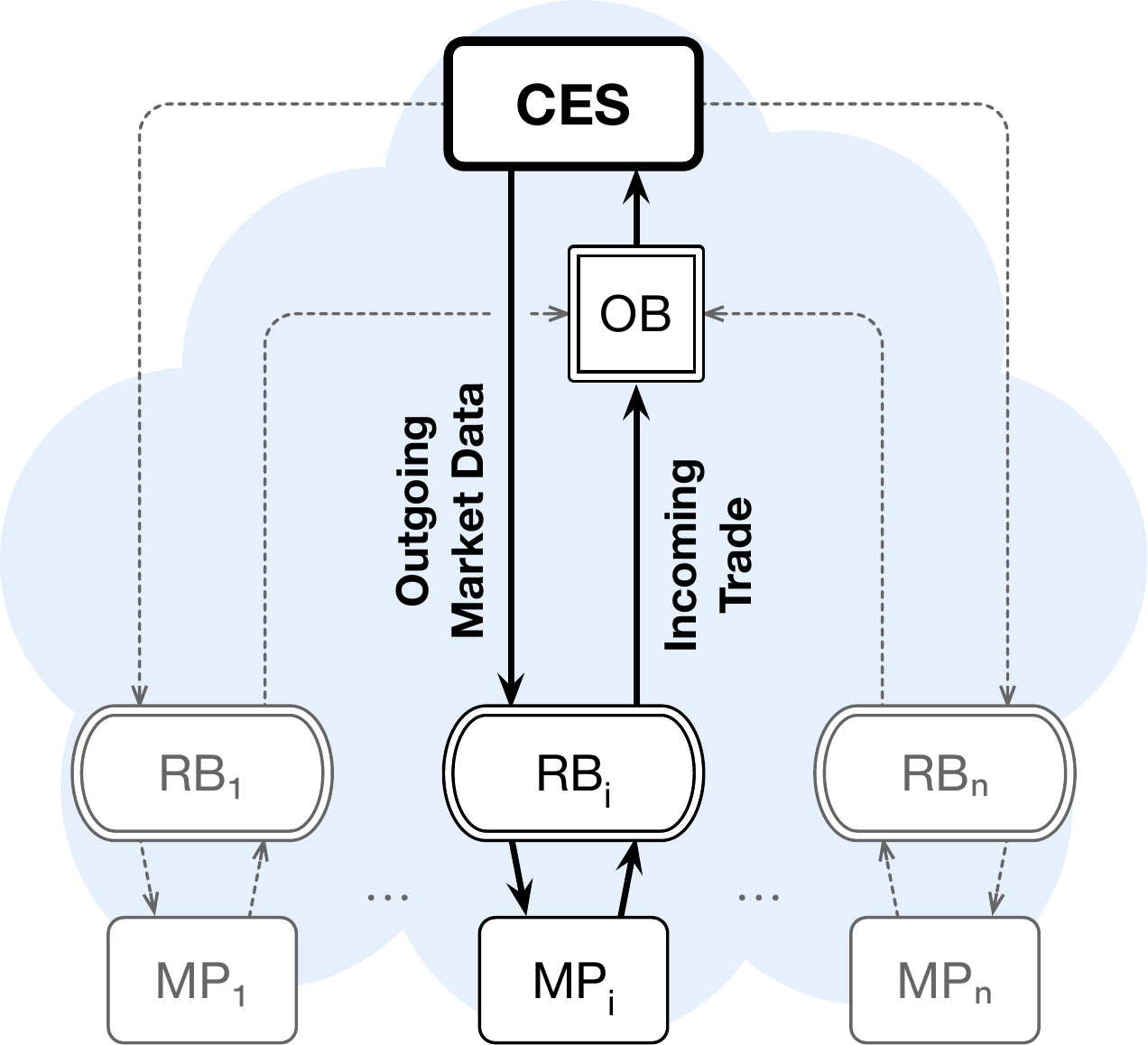}
     \vspace{-1mm}
    \caption{\small{\bf Basic components of \sys.}} 
    \label{fig:new_archtiecture}
    \vspace{-5mm}
\end{figure}

\subsection{Related Work}

The problem of moving financial exchanges to the cloud has received some attention. There are two bodies of work.

\noindent
%
\textbf{Clock-synchronization based solutions are not sufficient:} CloudEx~\cite{cloudex} proposes using clock synchronization to achieve an equal latency abstraction. CloudEx adds two new components to the architecture (as shown in Figure~\ref{fig:new_archtiecture}): (i) For each participant there is a colocated trusted component called the \emph{release buffer}, which buffers market data points, enabling a delayed delivery to the market participant. (ii) Likewise, the \emph{ordering buffer} at the CES buffers the trade order generated by the participants, enabling delayed and re-ordered delivery of the trade orders to the CES. The broader architecture of our solution is same as CloudEx.

In CloudEx, all the components have synchronized clocks. A market data point produced at time $t$ is released by the release buffers simultaneously at a prespecified time $t+C_1$. A trade order generated by the participant at time $t$ is forwarded to the CES by the ordering buffering at time $t+C_2$. The problem with this solution is that even with perfect clock synchronization if network latency spikes beyond the pre-specified thresholds, then such a system incurs unfairness. Cloud networks experience latency spikes that are a couple of orders of magnitude higher than the average. Latency spikes,  although rare,  are still unpredictable, and setting high thresholds can help guard against sudden spikes and help us achieve better fairness.  With such high thresholds though, the system incurs high latency even when the underlying network latency is well behaved. Safeguarding against tail spikes, increases the overall end-to-end latency ($C_1+C_2$) not only at the tail but on average as well.
More importantly, the main issue remains unsolved: 
still, there is no guarantee that equal bidirectional latency will always hold. In fact, there is a known impossibility result on this.

\noindent
\textit{Impossibility Result for equal bi-directional latency:} In network environments with finite but unbounded latency (common network model is distributed systems~\cite{lamportSeminalPaper}), even with perfectly synchronized clocks, it is impossible for two machines to communicate and co-ordinate to do a task at the same time (two generals problem \cite{two_generals}).
So two release buffers can never co-ordinate to deliver the same data to the respective market participants simultaneously,  no matter how they communicate with the CES or other release buffers. Note that it is still  to co-ordinate at the OB to ensure that latency on the reverse path stays the same (i.e., two trades generated at the same time are forwarded to the CES at the same time). 

\noindent\textit{Impossibility Result on Clock Synchronization:} Further, in network environments with unbounded network latency, it is also impossible to synchronize clocks to any extent and the error in clock synchronization is unbounded~\cite{imp_cs}. 

Our conversations reveal that exchanges wish to provide guaranteed fairness and as a result such solutions haven't seen much adoption.

\textbf{Modifying how the matching engine behaves:} Frequent Batch Auctions~\cite{frequent_batch_auctions} proposes releasing market data periodically in batches. The batch frequency is kept very low (1 batch per 100 ms) to allow all participants to respond before the next batch is released. All the trades corresponding to a batch are given the same priority for execution at the CES. This solution ensures fairness in the sense that no participant has an advantage over others because of network latency. However, the system latency is high (100 ms!). Further, this solution completely eliminates the speed races and a participant that responds to market data faster no longer has a competitive advantage. To achieve fairness in environments with unpredictable network latency, Libra~\cite{libra} assigns random priorities to the incoming trades. Libra achieves fairness for speed races stochastically (faster participants trades are ordered ahead more than 50\% of the times) when the variability in network latency is bounded.
Beyond the issues stated here, the main problem with both these solutions is that they require intrusive changes to the exchange matching algorithm. 


\section{Problem Statement}


\noindent
\textbf{Goals:} In this paper, we aim to solve the problem of enabling fair speed racing among high frequency traders in network environments where latency is finite but unbounded. We also do not wish to modify the matching engine to achieve this goal.
At a high level, our solution leverages the nature of the speed races to propose a new logical time domain -- delivery clocks -- that tracks time relative to when market data was delivered to the participants. By ordering trades using this delivery time domain we can achieve guaranteed fairness for such speed races.
Our goal here is not to just propose a solution, but also present theoretical insights that help researchers in understanding this space and enable future work.


\noindent
\textbf{Non-Goals:} Achieving bounded latency in cloud networks remains an open problem as of now. In this paper, we do not attempt to optimize the underlying network latency or the transport mechanism for multicasting market data or communicating trade orders. 
We also do not discuss solutions for reliability of the various components. Exchanges today incur unfairness in the event of failures \cite{signals_threads}. In our system, it should be possible to detect failure of various components and migrate the impacted components. During failures, fairness can get affected (\S\ref{ss:understanding_latency}).



We will now introduce some notation, formally define a speed race and fairness for such races. 

\noindent
\textbf{Notation:}
We refer to the $x^{th}$ market data point as $x$. $(i,a)$ refers to the $a^{th}$ trade from MP$_i$. Table\ref{tab:notation} lists the notations used in this paper.

\begin{table}[h!]
\small
    \centering
    \begin{tabular}{p{0.15\columnwidth} | p{0.75\columnwidth}}
        \textbf{Notation} & \textbf{Definition} \\
        \hline
        $G(x)$ & Real Time at which $x$ was generated at the CES.\\
        $D(i,x)$ & Real Time at which $x$ was delivered (by $RB_i$ in case of our system) to MP$_i$. \\
        $TP(i,a)$ & Market data point used to generate $(i,a)$.\\
        $RT(i,a)$ & Response time of $(i,a)$.  \\
        $S(i,a)$ & Real Time at which $(i,a)$ was submitted by MP$_i$.\\
        $F(i,a)$ & Real Time at which $(i,a)$ is forwarded (by OB in our system) to the CES's matching engine (ME). \\
        $O(i,a)$ & The order in which trades are forwarded (by OB in our system) to the CES. If $O(i,a) < O(j,b)$ then $F(i,a)$ < $F(j,b)$.

    \end{tabular}
    \caption{\small{Notation.}}
    \label{tab:notation}
    \vspace{-5.5mm}
\end{table}

\begin{figure}[t]
\centering
    \includegraphics[width=0.8\columnwidth]{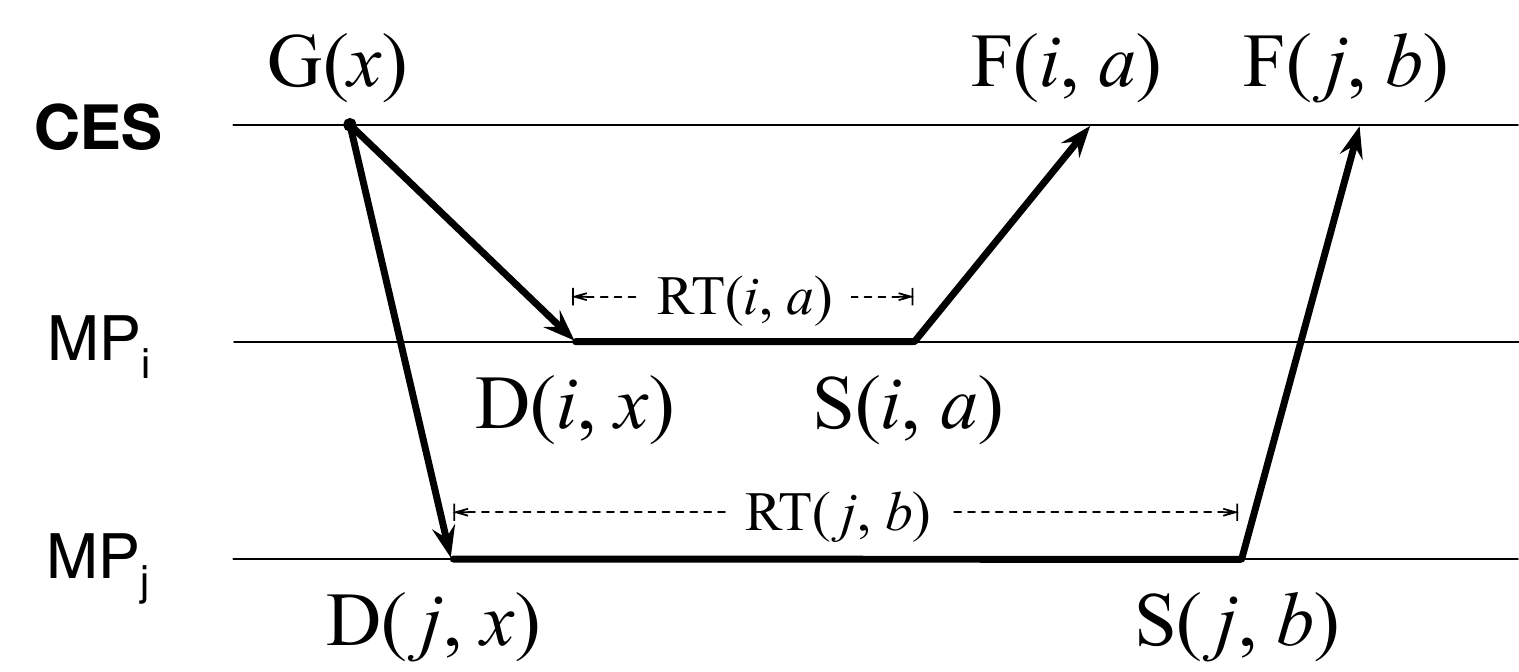}
    \caption{\small{\bf Events in a speed race.}}
    \label{fig:speed_race}
    \vspace{-2.5mm}
\end{figure}

 \noindent
 \textbf{Speed Race:} Informally, a speed race \cite{frequent_batch_auctions, libra, burdisch_working} consists of trades from multiple participants competing for the same trading opportunity. A particular market data point serves as the trigger/stimulus for trades competing in the speed race. Participants aim to identify the trading opportunity and win the speed race by responding as fast as possible after receiving the trigger market data point. 
 The trades belonging to a speed race are the most latency sensitive \cite{frequent_batch_auctions, libra, burdisch_working}. Differences in latency across participants in delivering the trigger point or on the reverse path to the CES can create significant disadvantages for certain participants~\cite{fragile, burdisch_working}. 
 These speed trades  constitute a substantial fraction of the overall trades in major exchanges (atleast 20\% in LSE \cite{burdisch_working}). In this paper, we will try to achieve fair ordering for trades engaged in such speed races.

\noindent
\textit{Compute model for speed trades:} The response time for trade $(i,a)$, $RT(i,a)$, whose trigger point is $x$ ($TP(i,a) = x$), is defined as the time it took to generate the trade after receiving the trigger point $x$. Formally, the time trade $(i,a)$ is submitted/generated by an MP is given by,

\begin{align}
    S(i,a) = D(i, x = TP(i,a)) + RT(i,a) 
    \label{eq:cm}
\end{align}

where $S(i,a)$ is the time trade $(i,a)$ is submitted/generated by an MP, and $D(i, x = TP(i,a))$ is the time at which $RB_i$ delivers $x$ to $MP_i$ (see Table 1).  

Response time captures the speed of the participant. Note that such a trade might be generated using market data points other than the trigger point. However, the trade submission time is completely governed by the delivery time of the trigger point and the response time of the participant for that trade.

 \noindent
 \textbf{Fair ordering of Trades in a Speed Race:} Outcome of a speed race is simply governed by the ordering of the competing trades in the race. Our goal is to achieve the same ordering for these trades had the network provided equal bi-direction latency. We refer to such an ordering of trades as \textit{Response Time Fairness}.

 In an equal bi-directional latency network ($C_1$ latency from CES to MP, $C_2$ latency from MP to CES), trade $(i,a)$ will be received by the CES at time,
 \begin{align}
     F(i,a) = G(x=TP(i,a)) + C_1 + RT(i,a) + C_2 
 \end{align}

By definition, Trade $(i,a)$ is ordered ahead of $(j,b)$, i.e., $O(i,a) < O(j,b)$), if $F(i,a) < F(j,b)$. In such a network, two trades $(i,a)$ and $(j,b)$ belonging to the same race (i.e. the same trigger point $x$) will be ordered as follows, 
\begin{align}
    \text{If } G(x) + C_1 + RT(i,a) + C_2 < G(x) + C_1 + RT(j,b) + C_2,\nonumber\\ \text{ then, } O(i,a) < O(j,b)
\end{align}

Using the above equation, we define response time fairness as follows,

\begin{definition}
An ordering system achieves response time fairness if it satisfies the following condition for all competing speed trades $(i,a)$ and $(j,b)$
\begin{align*}
    C1: &\text{ if } TP(i,a)= TP(j,b) = x\\ 
    &\land RT(i,a) < RT(j,b), \\
    &\text{ then, }O(i,a) < O(j,b).
\end{align*}
\label{def:rtf}
\vspace{-5mm}
\end{definition}

The above condition is simply stating that a faster participant's trades should be ordered ahead of slower participant. The above condition is from the perspective of the participants. Response time is not directly visible to the cloud provider or the exchanges. We will rewrite the above conditions using quantities visible to them.
The above condition can be rewritten as,

\vspace{-1mm}
\begin{align*}
    C1': \text{ if } &TP(i,a)= TP(j,b) = x \\
     &\land S(i,a) - D(i,x) < S(j,b) - D(j,x), \\
    \text{ then, } &O(i,a) < O(j,b).
\end{align*}

This condition states that the exchange can achieve response time fairness by measuring time of trades relative to when a market participant received the market data to order trades. 



Adding $G(x)$, i.e. the generation time of $x$, to both sides of the equation results in the following condition:

\begin{align*}
    C1'': \text{ if } &TP(i,a)= TP(j,b) = x \\
     &\land S(i,a) - (D(i,x)-G(x)) < S(j,b) - (D(j,x)-G(x)), \\
    \text{ then, } &O(i,a) < O(j,b).
\end{align*}

Here $D(i,x) - G(x)$  represents the one way latency from CES to participant $i$ for data point $x$.  
So to achieve response time fairness all the exchange needs to do is correct for the differences in latency from the exchange to the participant.

To deal with variability in network latency, CloudEx tries to equalize latency by holding information at the release buffer and releasing it simultaneously all participants using synchronized clocks. In other words, it strives to ensure that $(D(i,x)-G(x))$ is equal to $(D(j,x)-G(x))$, so that trades can simply be ordered by the time when they were submitted by the participants (i.e. $S(i, a)$). However, as discussed earlier, it is not possible to equalize latency always when the underlying network latency is unbounded. 


In this work, we take a different approach. Instead of trying to synchronize clocks or equalize latency (either of which can never be done precisely \cite{two_generals, imp_cs}), we show that it is possible to post facto correct for latency difference and achieve response time fairness.

\noindent
\textbf{Causality of trades from a participant:} We add an additional requirement for ordering of trades. This condition simply states that trades from a participant should respect causality, i.e, if trade $(i,a)$ was generated before trade $(i,b)$ then it should be ordered ahead. Formally, 
\begin{align}
\text{If } S(i,a) < S(i,b), \text{ then }, O(i,a) < O(i,b).
\label{eq:causality}
\end{align}

\noindent
\textbf{Fairness beyond Response Time Fairness:} While speed races are the most latency critical, in theory there can be latency-critical trades that don't fall under the speed race model (e.g., trades whose submission time depend on delivery time of multiple data points or some other external data). Guaranteeing perfect fairness for such trades does require simultaneous delivery of both market data and external data. While this is impossible, we will discuss how DBO can be enhanced to provide better fairness for such trades (\S\ref{ss:beyond_fairness}). 


 \noindent
 \textbf{Assumptions:} We will list out some of the assumptions we make in our solution.

\noindent
 \textit{Trust:} Release and ordering buffers are trusted components that are controlled by the cloud provider and that cannot be tampered with. 

\noindent
\textit{Proximity:} Release buffers are colocated with the participants. The latency between them is negligible. In our system, we implement the release buffer at participant's NIC. For scenarios where release buffer cannot be colocated we analyze the impact of latency between the release buffer and the participant. 

\noindent
\textit{Clock-drift rate:} We don't make any assumptions on clocks being synchronized across components in the network. Similar to literature in the traditional distributed systems~\cite{imp_cs}, we assume that clock-drift rate is negligible and release buffers can measure time-intervals accurately. Clock drifts rates are small in practice ($< 0.02\%$ under a wide range of scenarios~\cite{sundial}). 

\noindent
\textit{In-order delivery:} We assume that packets can be lost in the network. Packets that are not dropped are delivered in order. Just like exchanges today, we assume that all loses are handled out of band where the receiver requests retransmission using an alternative slower path~\cite{signals_threads}. Similar to modus operandi, our system incurs unfairness in such cases. 

\noindent
\textit{Participants are located in the cloud:} We assume that all the participants are located in the cloud. In case a certain participant doesn't want to move to the cloud, the exchange can run a proxy machine in the cloud on the behalf of such a participant. External participants can get market data feed and place trades through this proxy. Because of additional latency from proxy to the participant machines, trades from such external participants will be at a disadvantage. Fairness for other participants in the cloud remains unaffected.

\noindent
\textit{Remark:} CloudEx also makes the same assumptions on trust, proximity, and participants being in the cloud. The key difference is that CloudEx further assumes clock synchronization and requires bounded latency for guaranteeing fairness.

\subsection{Challenges}


There are three key challenges.

\noindent
\textbf{Challenge 1: Clock-synchronization:} Ideally, we want a solution that doesn't require any clock synchronization.



\noindent
\textbf{Challenge 2: Trigger point is unknown:} We assume that trigger point of a speed trade is not known. 
In such case it is hard to measure the response time and consequently decide how trades should be ordered. Unfortunately, when response times are unbounded it is impossible to achieve Response Time Fairness. 

\begin{theorem}
If trigger points for trades are unknown, then no ordering system can achieve Response Time Fairness.
\label{thm:1}
\end{theorem}

\begin{proof}
When trigger points are unknown, the ordering enforced by the system should achieve response time fairness for trades regardless what might have been their trigger point. This means that the ordering enforced by the system should respect condition $C1'$ regardless of what the trigger point $x$ is. The necessary condition for this to hold true is given below.

\begin{lemma}
When trigger points are unknown, the \textit{necessary} conditions on the delivery processes for achieving response time fairness with any ordering system is given by,
\begin{align*}
    D(i,y) - D(i,x) &= D(j,y) - D(j,x), & \forall i,j,x,y.
\end{align*}
\label{lemma:inter_delivery_imp}
\vspace{-6mm}
\end{lemma}

Please see Appendix~\ref{app:lem1} for proof. The lemma states that for response time fairness the inter-delivery times should be the same across all  participants. 
However, achieving the same inter-delivery time when network latency is unbounded is also impossible. If two processes can co-ordinate to achieve the same inter-delivery time then they can co-ordinate to do a task at the same time, a contradiction of the two generals impossibility result. 

 \end{proof}


We cannot achieve Response Time Fairness in settings where trigger points are unknown. We define a new slightly weaker version called \textbf{Limited Horizon Response Time Fairness} (LRTF) that is still useful. Formally, LRTF is defined as,

 \begin{definition}
An ordering system achieves limited-horizon response time fairness if it satisfies the following condition for all competing speed trades $(i,a)$ and $(j,b)$
\vspace{-1mm}
\begin{align*}
    C2: &\text{ if } TP(i,a)= TP(j,b) = x\\ 
    &\land RT(i,a) < RT(j,b), \\
    & \land RT(i,a) < \delta,\\
    &\text{ then, }O(i,a) < O(j,b).
\end{align*}
\label{def:lrtf}
\vspace{-5mm}
\end{definition}

The above condition is similar to condition (C1) with an additional constraint that the system guarantees response time fairness for only fast trades that are generated within a bounded amount of time. 
Notice that the constraint on response time being less than $\delta$ is only on  participant $i$.  Participant $i$'s trades will be ordered fairly regardless of whether the response time of other participant's competing trades is within $\delta$ or not.
\emph{In this paper, we will present a system, DBO, that for any given $\delta$  achieves LRTF in a guaranteed manner.}


\textit{Why is LRTF useful?} LRTF is based on the fact that typically participants response very quickly to market data. From our conversations, the faster participants in major exchanges responds within a few microseconds. \cite{burdisch_working} further shows that majority of the speed races last 5-10 microseconds. An exchange provider can choose to offer its participants guaranteed response time fairness for fast trades. The choice of $\delta$ does present a trade-off, increasing $\delta$ increases system latency.

\noindent
\textbf{Challenge 3: Enforcing the ordering} Suppose we could mark the trades at generation with the sequence in which they should be forwarded to the CES. Because trades can take unbounded amount of time on the reverse path, even in this scenario its hard to enforce such an ordering at the ordering buffer. In particular, before forwarding trade $(i,a)$ we need to be sure that there is no other trade $(j,b)$ in flight that should be ordered ahead.

\section{Design}
\label{s:design}
In this section, we will first present the core of our system. Then we present some analysis of the system along with some extensions to address a few practical concerns. We will present details of our cloud implementation separately in the next section.

\subsection{Delivery Based Ordering}
Our solution is composed of three parts. 
\subsubsection{Delivery Clock\\}
\noindent\textbf{What we do.}
Each RB maintains a delivery clock. This delivery clock essentially tracks time relative to when market data was delivered to the participant. We use $DC(i,a)$ to represent delivery clock of participant $i$ at time when trade $(i,a)$ is submitted. Delivery clock is a lexicographical tuple.
\begin{align}
    DC(i,a) = \langle ld(i,a), S(i,a)-D(i, ld(i,a))\rangle.
\end{align}
where $ld(i,a)$ is the latest data point that was delivered to MP$_i$ by time S(i,a), i.e., $D(i,ld(i,a)) \leq S(i,a) < D(i,ld(i,a)+1)$). 
Interval, $S(i,a)-D(i, ld(i,a))$, corresponds to the time that has elapsed since the last delivery and can be measured locally at the RB without requiring any clock synchronization (challenge 1). 

\noindent
\textit{Monotonicity:} Delivery clocks advance monotonically with submission time. As a result, DBO trivially satisfies the causality condition (Equation~\ref{eq:causality}). Further, it incentivizes the participants to submit trades as early as possible. Therefore, \emph{a participant cannot gain any advantage by delaying trades.} 
Finally, we also leverage the monotonic property to overcome challenge 3 (\S\ref{ss:enforcing_ordering}). Figure~\ref{fig:delivery_clock} shows how delivery clock advances with time.


\begin{figure}[t]
\centering
    \includegraphics[width=0.8\columnwidth]{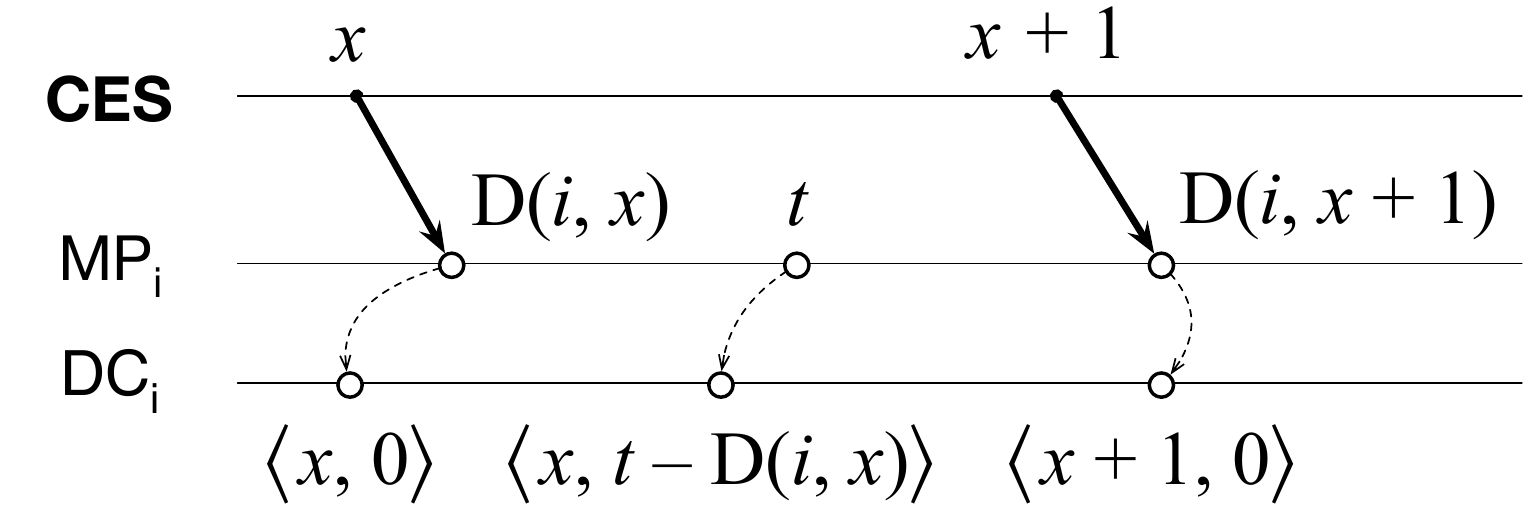}
    \caption{\small{\bf Delivery Clock.}}
    \label{fig:delivery_clock}
    \vspace{-2.5mm}
\end{figure}

All incoming trades are marked with the delivery clock at the trade submission time. The ordering buffer uses this delivery clock time to order trades. Formally, the ordering in DBO is given by,  

\vspace{-2mm}
\begin{align}
    O(i,a) = DC(i, a). 
    \label{eq:ordering_with_dc}
\end{align}

\begin{figure}[t]
\centering
    \includegraphics[trim={0 0 0 2mm},clip,width=0.8\columnwidth]{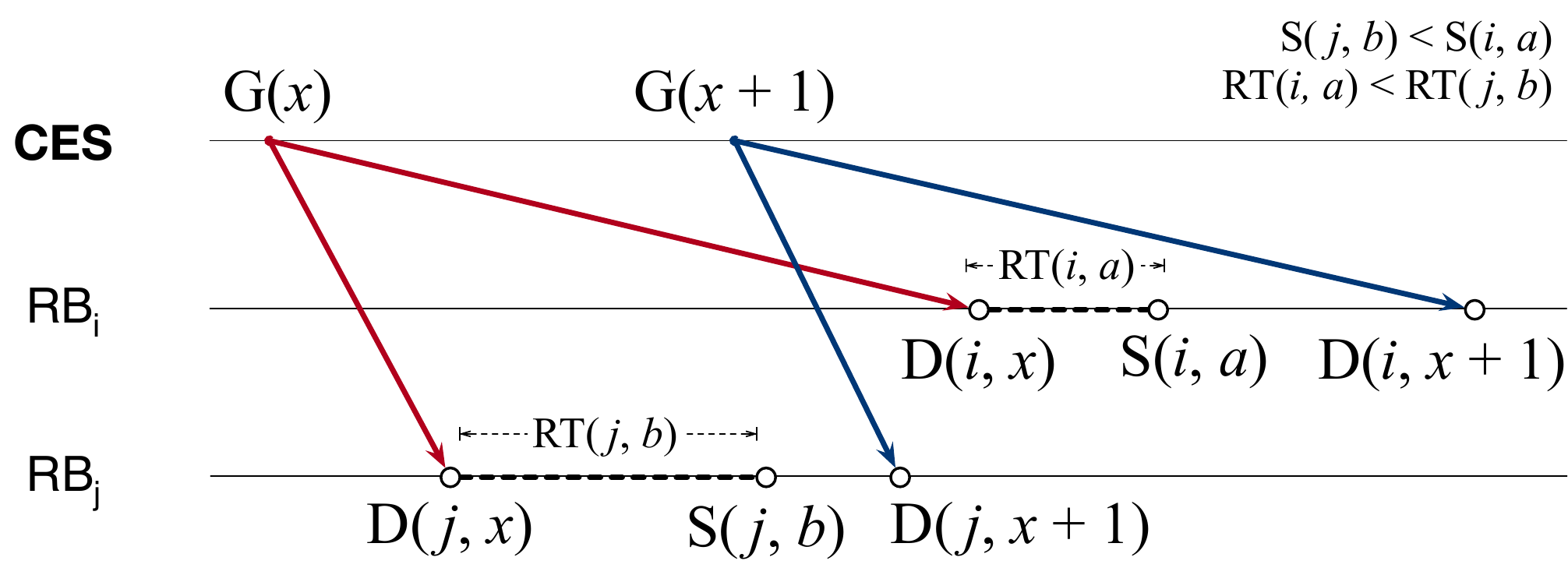}
    \vspace{-4mm}
    \caption{\small{{\bf DBO can help correct for late delivery of data.} Delivery of market data to MP$_i$ is lagging behind MP$_j$. There are two trades $(i,a)$ and $(j,b)$ generated in response to the same market data $x$. $(j,b)$ was submitted before $(i,a)$ but
    response time of $(i,a)$ is less than $(j,b)$.
    In this example, $DC(i,a) (= \langle x, RT(i,a)\rangle) < DC(j,b) (= \langle x, RT(j,b)\rangle)$ and trade $(i,a)$ is correctly ordered ahead of $(j,b)$.}} 
    \label{fig:dbo_correction}
    \vspace{-3mm}
\end{figure}

\noindent\textbf{Why it works.}
When the trigger point of trade $(i,a)$ is indeed the last data point (i.e., $x = TP(i,a) = ld(i, a)$), then, DBO respects condition C2 for LRTF. Figure~\ref{fig:dbo_correction} shows an illustrative example of this.
This is because, the delivery clock directly tracks the response time of $i,a$ in this case and $O(i,a) = DC(i, a) = \langle x, RT(i,a)\rangle$. For a competing trade $(j,b)$ with higher response time, the delivery clock at time of submission will either read $O(j,b) = DC(j, b) = \langle x, RT(j,b)\rangle$ (if S(j,b)<D(j,x+1)) or $DC(j, b) = \langle y, S(j,b)-D(j,y)\rangle$ with $y>x$. In both cases, $O(i,a) < O(j,b)$.

At a high level, in our ordering we are correcting for latency differences in data delivery by using the delivery time of the last data point. When the last data point is not the trigger point for trade $(i,a)$, DBO satisfies the LRTF condition C2, if the following condition holds, 
\begin{align}
    D(i,ld(i,a))-D(i,x) = D(j,ld(i,a))-D(j,x),
    \label{eq:cond_delivery_lrtf}
\end{align}
where $x = TP(i,a)$.  
While it is impossible to ensure that inter-delivery times remain the same for all participants for all points, by pacing data at the RB it is indeed possible to ensure that the above condition is always met.
The main reason why we can meet the above condition is that condition C2 limits that the trigger point $x$ cannot be any arbitrary data point in the past, and that the trigger point must have been delivered recently  $S(i,a)-D(i,x) < \delta$.
In the next subsection, we will show how we can achieve this and solve challenge 2. 


\noindent
\textit{Remark:} In our cloud experiments, we find that DBO achieves fairness with very high probability. This is because network latency (from CES to any given participant) exhibits temporal correlation in latency especially over  short periods of time. When temporal correlation is high, inter-delivery time at any participant is close to the inter-generation time at the CES. In such cases, condition given by Equation~\ref{eq:cond_delivery_lrtf} is satisfied with high probability.

\noindent
\textbf{Difference with traditional logical clocks:} Logical clocks are commonly used in distributed systems. The most famous ones are lamport clocks~\cite{lamportSeminalPaper} and vector clocks. These clocks can be used for achieving total causal ordering of events. While these clocks can track causality of events, they cannot be used to achieve response time fairness. In particular, these clocks don't say anything about how two competing trades generated using the same market data should be ordered as these two trades have no direct causality relation. Unlike delivery clocks, such logical clocks also have no notion of measuring time between occurrences of two events. Time difference between events is critical to achieve fairnesss for exchanges. 

\noindent\textit{Note:} Several major financial exchanges already rely on heartbeats~\cite{nyse-client} for liveness when traffic is low.

\begin{figure}[t]
\centering
    \includegraphics[width=0.8\columnwidth]{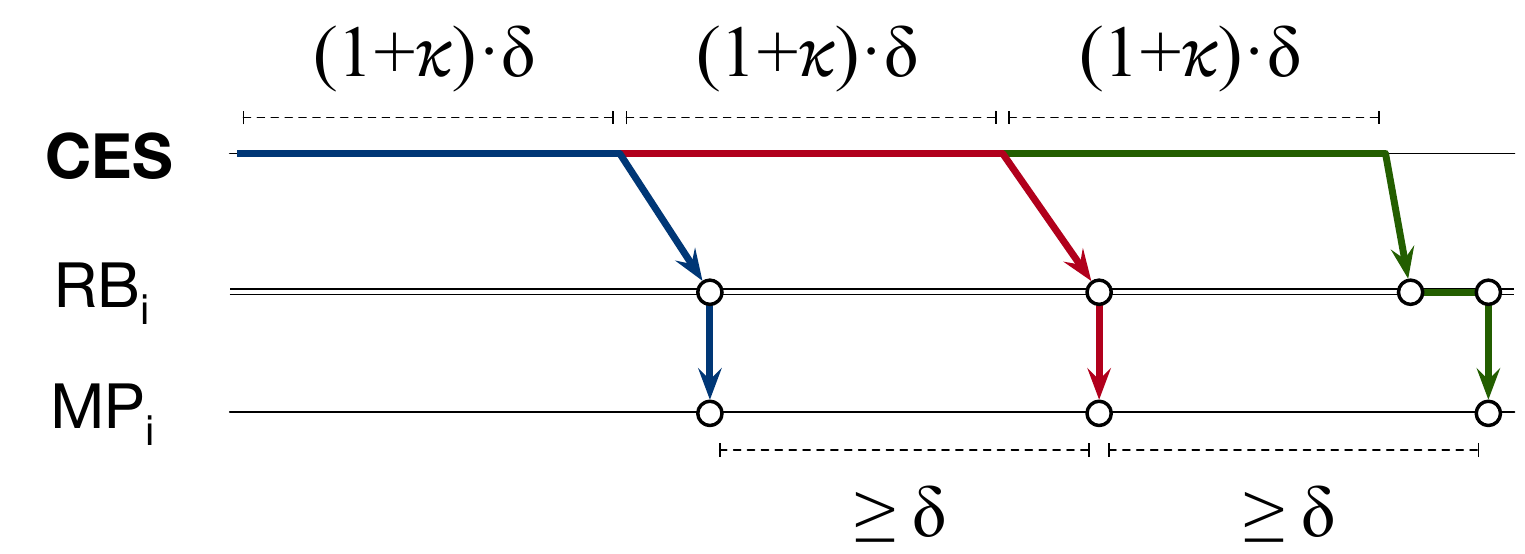}
    \vspace{-2mm}
    \caption{\small{\bf Batching and Pacing. Inter-delivery time for consecutive batches is equal to or more than $\delta$.}}
    \label{fig:batching_pacing}
    \vspace{-4.5mm}
\end{figure}

\subsubsection{Batching and Pacing\\}
\noindent
\textbf{What we do.}
In DBO, the CES breaks data into batches. Each new batch contains all data points in the duration $(1+\kappa) \cdot \delta$ after the previous batch. Here $\kappa > 0$. Each release buffer delivers all data points in a batch at the same time. 
The release buffer delivers batches as quickly as possible while ensuring that the time between delivery of two consecutive batches is atleast $\delta$. Figure~\ref{fig:batching_pacing} shows an illustration of batching. Both batching and pacing increase the delivery time of data points. In the next subsection we will analyze the impact of the two on latency. Note that in the event of queue build up at the RB, since the batch generation rate ($\frac{1}{(1+\kappa) \cdot \delta}$) is slower than the batch dequeue rate($\frac{1}{\delta}$), the queue at the RB eventually gets drained(\S\ref{ss:understanding_latency}).

\noindent
\textbf{Why it works.} With batching and pacing, DBO achieves LRTF. In particular, 
consider a trade $(i,a)$ with response time less than $\delta$. Because of pacing, consecutive batches are separated atleast by $\delta$. This means that the trigger point ($x=TP(i,a)$) must be within the last received batch. The point $ld(i,a)$ is also the last point in this batch and $D(i,ld(i,a)) = D(i,x)$. \emph{With batching and pacing, the delivery clock again directly tracks the response time of $(i,a)$} and $O(i,a) = DC(i,a) = <ld(i,a), RT(i,a)>$.
With batching, for participant $j$, $x$ and $ld(i,a)$ also belong to the same batch $D(j,ld(i,a)) = D(j,x)$.
For a competing trade $(j,b)$ with higher response time, the delivery clock at the time of submission will either read $O(j,b) = DC(j,b)) = \langle ld(i,a)), RT(j,b)\rangle$ (if $(j,b)$ was submitted before the next batch, i.e., $S(j,b) < D(j,ld(i,a)+1)$) or $DC(j, b) = \langle y, S(j,b)-D(j,y)\rangle$ with $y>ld(i,a)$. In both cases, $O(i,a) < O(j,b)$.

\subsubsection{Enforcing the ordering\\}
\label{ss:enforcing_ordering}
OB contains a priority queue where all incoming trades are sorted based on the delivery clock timestamp (Equation~\ref{eq:ordering_with_dc}). A trade $(i,a)$ at the head of the priority queue should be forwarded to the CES only when the OB has received all trades $(j,b)$ with lower ordering $DC(j,b) < DC(i,a)$. 

\noindent
\textit{OB's Heartbeat Handler:} In DBO, each RB sends a heartbeat periodically every $\tau$ seconds to the CES. The heartbeat $(i,h)$, from participant $i$ contains the delivery clock timestamp at the time the heartbeat was generated ($DC(i,h)$). Since data in delivered in order and because delivery clock advances monotonically with time, heartbeat $(i,h)$ tells the OB that it has received all trades from participant $i$ with delivery clock less than $DC(i,h)$. The ordering buffer forwards trade $(i,a)$ if it has received heartbeats from all the participants with delivery clock timestamp higher than $DC(i,a)$.

\subsection{Understanding DBO}

\subsubsection{Latency, parameter setting and straggler mitigation\\}
\label{ss:understanding_latency}

We will first derive the optimal latency for any ordering system that achieves response time fairness. We will then discuss how DBO compares to  optimal latency. We will also present guidelines for setting parameters and how to mitigate stragglers that can impact latency.

We define latency for trade $(i,a)$, $L(i,a)$, as the sum of latency in delivering data (from generation time) and latency in trade forwarding to the CES (from trade submission time). Formally,
\begin{align}
    L(i,a) = (D(i, x) - G(x)) + (F(i,a) - S(i,a)),\nonumber\\
    L(i,a) = F(i,a) - G(x) - RT(i,a),
    \label{eq:latency_def}
\end{align}
where $x=TP(i,a)$.

\noindent
\textbf{Optimal Latency:} Formally trade $(i,a)$ can only be forwarded to the CES's ME only when the CES has received all potential competing trades $(j,b)$ with lower response times ($RT(j,b) < RT(i,a)$). Let $R(i, x, RT)$ represent the time when the CES receives trade $(i,a)$ whose whose trigger point is x and response time is RT. Formally, 
\begin{align}
    F(i,a) = \max_{j}(R(j, x=TP(i,a), RT=RT(i,a))). 
\end{align}
A subtle point to note here is that even if participant $j$ does not produce any trades, we still need to wait for that participant till $R(j, x=TP(i,a), RT(i,a))$. Before this time, fundamentally the CES cannot be sure that it will not receive a trade from participant $j$ with a lower response time. 

We use $RTT(i, x, RT)$ to represent the sum of raw network latency for point x from CES to MP $_i$ and latency of trade from MP$_i$ to the CES (whose trigger point is x and response time RT).  In the best case scenario for latency (no buffering at any point in the path) we get
\begin{align}
    R(i, x, RT) = G(x) + RTT(i, x, RT) + RT.
\end{align}

Using the above two equations, we can write the following theorem.
\begin{theorem}
For any ordering system that achieves response time fairness, the minimum latency for trade $(i,a)$ is given by,
\begin{align}
    L(i,a) = \max_{j}(RTT(j, x=TP(i,a), RT=RT(i,a))).
\end{align}
\vspace{-2mm}
\label{thm:latency}
\end{theorem}

Put it simply, the above theorem states for achieving response time fairness, the minimum latency is bounded by the maximum round trip time across all participants. This means that fundamentally bad latency for a participant affects the latency of all trades. To achieve low latency consistently, we would like to ensure that latency of all the participants is well behaved majority of the times. How to better achieve this goal is left as a subject for future work.



\noindent
\textbf{How does DBO compare with the optimal?} DBO achieves close to optimal latency.  Compared to the optimal, batching and pacing introduce additional delay in delivery of market data points.  Since heartbeats are  generated only periodically they can  introduce an additional delay of $\tau$ at the ordering buffer. We now discuss the delay due to each of these components and how do the parameters $\kappa$, $\delta$ and $\tau$ affect latency. 

\noindent
\textbf{Impact of batching:} Batching can introduce an additional delay of $(1+\kappa)\cdot \delta$ in the worst case. 

\noindent
\textit{Setting $\delta$:} $\delta$ thus presents a trade-off between latency and fairness (how large of a horizon can we pick). The right trade-off really depends on the needs of the exchange. Ideally, the exchange should pick the minimum value of $\delta$ that accommodates the response time of the fastest participants in a race. Our conversations reveal that fastest participants typically respond within a few microseconds and majority of the speed races last 5-10 $\mu s$. For our cloud experiments we  use $\delta = 20 \mu s$.

\begin{figure}[t]
    \centering
    \includegraphics[trim={0 0 0 0mm},clip,width=0.8\linewidth]{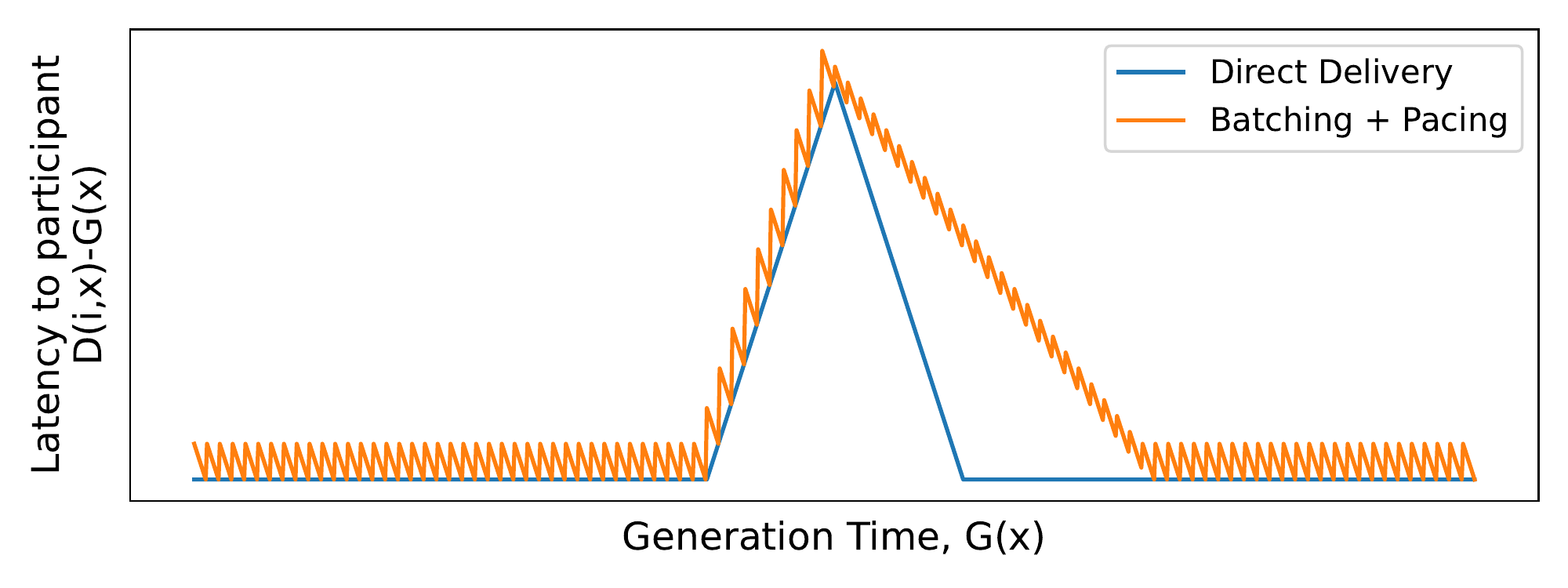}
    \vspace{-5mm}
    \caption{\small{\textbf{Latency in data delivery:} x-axis shows the generation time of the market data. y-axis shows the latency from generation time to data delivery. $\kappa$  governs the average slope of the orange line immediately after latency spike (slope = $\frac{\kappa}{1+\kappa}$}).} 
    \label{fig:latency_b+p}
    \vspace{-5mm}
\end{figure}

\noindent
\textbf{Impact of pacing.} Pacing restricts the batch dequeue rate at the RB. When network latency to a participant is not varying, the batch arrival/enqueue rate at the RB ($\frac{1}{(1+\kappa) \cdot \delta}$) is higher than the batch dequeue rate limit ($\frac{1}{\delta}$) and there is no queue build up. However, when network latency to a participant is decreasing (e.g., after a latency spike), batch arrival rate at the RB can exceed the dequeue rate limit leading to a queue build up. The overall queue - dequeue rate can be given by $\text{batch size} \cdot \text{batch rate limit} = 1+\kappa$. Figure~\ref{fig:latency_b+p} shows the impact of batching and pacing on latency in delivery of data in the event of a queue build up. The figure also shows the latency when data is delivered directly (raw network latency). The smaller sawtooths in the batching + pacing are because of batching. The deviation in direct delivery and batching + pacing is because of the rate limit imposed by pacing.

\noindent
\textit{Setting $\kappa$:} Increasing $\kappa$ increases batching delay but also increases the queue drain rate in the event of queue build up due to tail latency spikes. Increasing $\kappa$ thus presents a trade-off between reducing tail latency and increasing average latency. In our experiments we use $\kappa = 0.25$.
 
\noindent\textbf{Impact of heartbeats:} Heartbeats present a trade-off. Too frequent heartbeats can overwhelm the network, the ordering buffer or the release buffer. 
Infrequent heartbeats can increase the time OB has to wait of the participants. In particular, hearbeats can introduce an additional wait time of $\tau$. Note that the number of heartbeats, the OB needs to process increases linearly with the number of participants. In the next section we show how the heartbeat handler can be sharded for scalability.

\noindent\textit{Setting $\tau$:} Ideally we want to pick as low of a value as possible for the heartbeats without overwhelming the system. This number is very much dependent on the capabilities of the network and the processing power of the RB and the OB. In our cloud implementation we use $\tau = 20 \mu s$.

\noindent\textit{A note on latency:} When the network latency to participants is not varying with time, there is no queue build up at the release buffers. In such cases, DBO adds maximum of $((1+\kappa)\cdot \delta) + \tau$ additional latency over the optimal.

\noindent\textbf{Straggler Mitigation and RB/MP failure} In the event a  participant or release buffer crashes, DBO can stall processing trades. Further, the overall system latency also gets impacted when a certain participant is experiencing unusually high network latency (see Theorem~\ref{thm:latency}). Here we have the option to wait for the delayed participant and take a latency hit but not let the fairness be impacted. Ideally, we want to let the system continue with low latency with only the affected participant incurring unfairness. In DBO, we use a simple strategy to mitigate this. Using the heartbeats and the generation time of data points, the OB tracks the round trip latency to each participant. If this latency goes beyond a certain threshold for a participant, then the OB does not wait for heartbeats from such straggler participant before forwarding trades. When the round trip latency goes down, OB again starts waiting for heartbeats from the straggler. In the event of crashes, OB might not hear any heartbeats. If the OB does not hear a heartbeat from a particular participant for the above threshold, then it concludes that round trip latency exceeds the threshold and the OB deems the participant a straggler. 
 
\noindent\textit{OB failure:} In the event, the OB crashes all trades in the priority queue will be lost. System will incur unfairness in such cases. 


\subsubsection{Is Batching and Pacing necessary?\\}
\textbf{Batching and pacing contribute delays; are they necessary?} The answer is yes. Similar to Lemma~\ref{lemma:inter_delivery_imp}, we can derive the necessary conditions for achieving LRTF. 
\begin{corollary}
When trigger points are unknown, the \textit{necessary} conditions on the delivery processes for achieving response time fairness with any ordering system is given by,
\vspace{-1mm}
\begin{align*}
    \text{If }  D(i,y) - D(i,x) &< \delta, \text{ then},\nonumber\\
    D(i, y) - D(i,x) &= D(i,y) - D(i,x), & \forall i,j.
\end{align*}
\label{cor:inter_delivery_lrtf}
\vspace{-6mm}
\end{corollary}

\begin{proof}
Please see Appendix~\ref{app:cor_inter_delivery_lrtf}.
\end{proof}
\vspace{-1mm}
In contrast to Lemma~\ref{lemma:inter_delivery_imp}, the above condition states that the inter-delivery time of two points should be same across all participants only if they are separated by less than $\delta$ for some participant. Batching and pacing indeed satisfies this, for two points x and y in a batch, the inter-delivery times across all participants is indeed zero and hence equal. For point $x$ and $y$ belonging to different batches, since the inter-delivery time is greater than $\delta$ across all participants, there is no additional contraint on inter-delivery times being equal.
 
\subsubsection{Impact of RB to MP latency\\}
In scenarios where RB and the participant cannot be colocated, DBO can incur unfairness. If this latency is unbounded, then, it might be impossible to achieve fairness. If latency is bounded, however, then DBO provides the following fairness guarantees.

\begin{theorem}
    If round trip network latency from release buffer $i$ to it's corresponding participant is bounded between $B_l(i)$ and $B_h(i)$, then, DBO achieves the following guarantee for ordering trades.
    \begin{align*}
    C3: &\text{ if } TP(i,a)= TP(j,b) = x\\ 
    &\land RT(i,a) < RT(j,b) - (B_h(i)-B_l(j)), \\
    & \land RT(i,a) < \delta - B_h(i),\\
    &\text{ then, }O(i,a) < O(j,b).
\end{align*}
    \label{thm:rb_to_mp_latency}
    \vspace{-5mm}
\end{theorem}

\vspace{-1mm}
\begin{proof}
See Appendix~\ref{app:rb_to_mp_latency}.
\end{proof}
\vspace{-1mm}

Compared to LRTF, the above condition reduces the bound on response time for the faster trade $(i,a)$ to $\delta - B_h(i)$.
Additionally, the above condition states that trades are ordered fairly only when the response time of the faster trade is lower than the response time of the competing trade by atleast the variability in latency ($B_h(i)-B_l(j)$). This theorem essentially states that when RB and MP cannot be colocated, for better fairness we should ensure that latency between them is both consistent (across participants) and the upper bound is small.

\subsubsection{Impact of Losses\\}

Although infrequent, packet losses can occur in cloud environments. Such losses can impact fairness in DBO. However, only the fairness for trades that are lost and trades  whose trigger point is lost is impacted (see Appendix~\ref{app:impact_losses}).

\subsubsection{Thwarting front-running attacks\\}

There is a front-running attack possible in our system. In particular, if a participant receives a market data point $x$ through some other way before RB delivers the data point $x$ to the participant then the participant has a competitive advantage. This scenario (though unlikely) is still possible. 

A simple to avoid this is to limit that a participant cannot talk to anyone beyond the CES. 
However, we would like the participant machine to use other  ``helper'' machines in the cloud, e.g.,  to aid computation. We also want to allow the participants to be able to talk to machines outside the cloud, e.g., to get a news stream. 

In Appendix~\ref{app:front_running}, we show how we can prevent such front running attacks. In our solution, the participant and its helpers cannot communicate with any other participants or their helpers using the cloud network. 
To prevent scenarios where a participant uses a proxy machine outside the cloud to send market data to other  participants (faster than the network), we precisely add additional latency for data being sent outside the cloud.
While our solution introduces latency for data going out, the latency of speed trades remains unaffected.

\subsubsection{Limtations of DBO: Fairness beyond LRTF\\}
\label{ss:beyond_fairness}

With DBO, it is not guaranteed that trades that do not directly follow the LRTF model (Theorem~\ref{thm:1} and Equation~\ref{eq:cm})are ordered fairly. However, DBO still ensures that fairness for the most latency-sensitive speed trades. While ensuring guaranteed fairness for trades that do not follow the might be impossible, we will discuss potential some solutions.

%

\noindent\textbf{Trades with response time > $\delta$:} DBO does not provide any guarantees for trades with response time greater than $\delta$. 
In case we have access to synchronized clocks, we can try and ensure (to the extent possible) that batches are indeed delivered at the same time across participants. 
When batches are delivered simultaneously, delivery clocks also get synchronized and DBO simply orders trades in the order of submission time. DBO thus ensures better fairness for such trades (when data is delivered simultaneous) while always guaranteeing LRTF. 


\noindent\textbf{Generalized compute model for trades:} A trade's submission time might be governed by delivery times of multiple data points. Again in such cases if we have access to synchronized clocks, we can try and ensure simultaneous delivery to the extent possible and achieve better fairness for such trades.

\noindent\textbf{External data streams:} In theory, external data streams like news events or market data from a competing exchange can trigger speed races. While DBO does not delay delivery of such streams to the participants (Appendix~\ref{app:front_running}), as described it does not guarantee fairness with respect to such streams. Existing exchanges do not provide any simultaneous delivery guarantees with respect to such external streams. Such streams typically traverse the internet, and the variability is network latency is substantially higher (order of milliseconds) than the market data stream (order of microseconds). Potentially, the exchange can serialize such external streams with the market data stream and ensure LRTF with respect to such a super stream. Such a serialization might not be trivial. Participants are requesting different data streams. We need to think carefully about what constitutes a fair serialization.
\section{Cloud Architecture and Implementation}
\label{s:cloud_arch_impl}

In a typical on-premise deployment, the CES servers and physical network are part of the trusted infrastructure of the exchange: the exchange operators have exclusive access to the physical machines, network elements and cables. On the other, the MPs own the physical servers that connect to the exchange network. Migrating such components to the public cloud is slightly more complicated: while the CES servers and MPs could be accommodated by virtual machines owned by the different parties, but the network infrastructure is still owned by the cloud provider. Compared to on-premise deployments, \sys requires leveraging two extra components for correctness: the Release Buffer (RB) and the Ordering Buffer (OB). 

\subsection{Release Buffer}
\label{ss:release_buffer}


Figure~\ref{fig:rb_archtiecture} depicts a high-level view of the RB's functionality. The RB transparently interposes the communication between the Market Participant (MP) and the Ordering Buffer (OB). As mentioned previously, the RB maintains the Delivery Clock (DC),  the logical clock tuple consisting of id of the latest data point transmitted to the MP and the time elapsed since the last transmission. Market data is grouped into logical batches 
by the CES and sent to each MP. The RB buffers the received market data (packets) that belong to the same batch, until the full batch is received. 
Upon the reception of the last market data (packet) of the batch, the RB checks whether the time elapsed since the previous market data batch delivery to the MP: if it is equal to or more than $\delta$, the batch is released to the MP at once, and the DC is updated on transmission completion of each packet. Otherwise, the batch is buffered at the RB for the appropriate duration to ensure that inter-batch gap is equal to or more than $\delta$.

 Each MP implements its own strategy on how to respond to each market data received, and generates trades. All the trades from an MP are intercepted by the correspondent RB: upon the reception of a trade, the RB needs to tag the trade accordingly with a DC-derived timestamp so that full ordering can be achieved at the Ordering Buffer. This timestamp is piggybacked on each trade and is calculated simply as the tuple consisting of the current DC id and the real time elapsed between trade reception and latest market delivery.

\begin{figure}[t]
\centering
    \includegraphics[width=0.85\columnwidth]{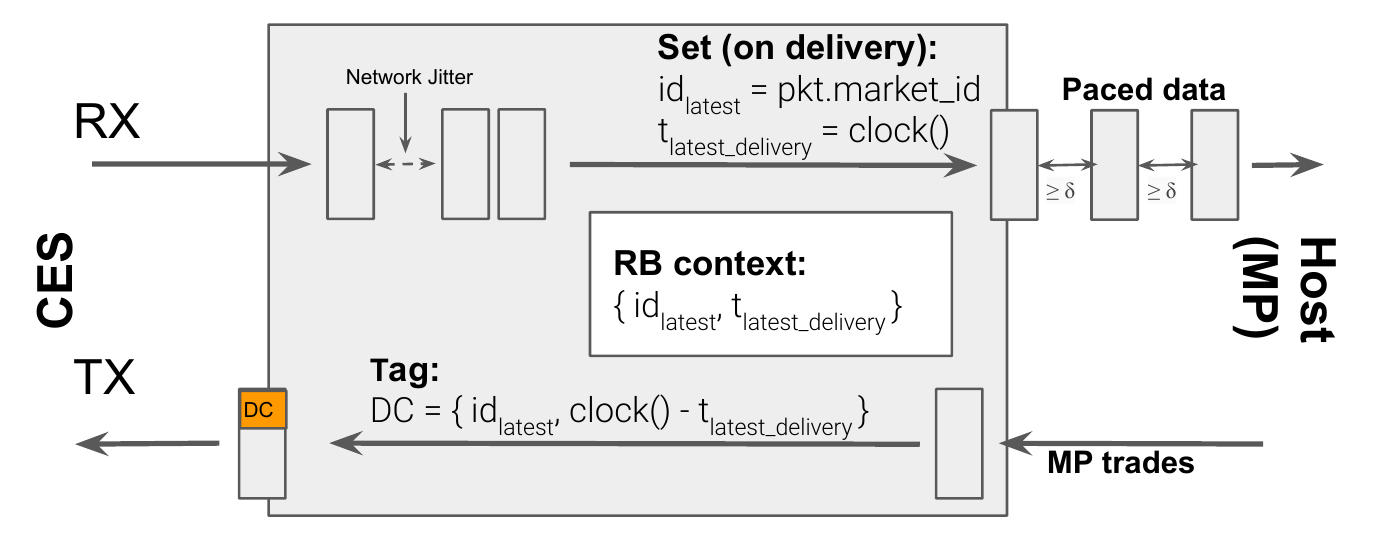}
    \caption{\small{\bf High-level architecture of the Release Buffer. The Delivery Clock advances upon new market data reception from the CES. Incoming trades from the MP are tagged with the Delivery Clock id and MP's response time before sent to the OB/CES.}}
    \label{fig:rb_archtiecture}
    \vspace{-5mm}
\end{figure}

Where should the RB placed in a cloud-hosted Financial Exchange deployment? There are two essential requirements for the RB component: a. the latency between MP and RB must be minimal so that it does not affect correctness, and b. for security reasons, the RB must be isolated from the MP, to avoid attacks that aim to tamper with response time measurements or market data delivery. Deploying the RB as a standalone VM is not a solution, as that would introduce non-negligible, variable latency between MP and RB. Even for VMs that are collocated into the same physical node, inter-VM communication is still achieved via network communication so that cloud providers can enforce the appropriate SDN policy. A switch-based implementation would also suffer from similar limitations: a. there is lack of fine-grained control for VM placement in cloud (so we cannot have any guarantees about switch-VM latencies), b. switch resources are scarce and shared by multitenant traffic  in the cloud, and avoiding interference would be a challenging problem to solve. 


Top-tier cloud providers deploy (custom) programmable NICs that leverage a variety of ASIC- or FPGA-based accelerators and powerful SoCs to enforce strict SDN policies required for I/O resource management, network virtualization, billing etc. These platforms serve as a natural boundary between the guest VMs that are controlled by the customers and the datacenter network which is shared resource managed by the cloud operator. We believe that the RB's functionality should be embedded in the cloud providers' smartNICs. RB support in the programmable NIC could be incrementally deployed in the existing infrastructure, and exposed to customers as a virtual NIC feature similar to accelerated networking \cite{firestone2018azure, efa, ena}. NIC performance isolation and background interference challenges are beyond the scope of this paper: MPs already invest large amounts of money for their co-located server hardware -- using high-end instances that provide single-tenancy guarantees per cloud node (cite dedicated instances) would eliminate interference stemming from on-host multi-tenancy. 

Since we do not have access to cloud providers' smartNICs, we used an off-the-shelf programmable (DPU) NIC~\cite{bf2} to demonstrate the feasibility of a NIC-based RB implementation. We implemented the RB functionality on top of DPDK~\cite{www-dpdk}, running it on the System-on-Chip ARM cores. A busy-polling receive engine intercepts all incoming market data traffic and releases them to the host while enforcing the pacing requirements. The RB functionality is completely transparent for the MP: market data  packets appear at the host's RX ring unmodified.

\subsection{Ordering Buffer}
\label{ss:ob}

The Ordering Buffer component's functionality closely resembles that of a `sequencer' which tags incoming trades in a First-Come-First-Served (FCFS) manner in existing on-premise deployments. In our system, it is responsible for ordering all received trades based on their Delivery Clock timestamp, before they are submitted to the Matching Engine (ME). Similarly to the `sequencer', the OB component is part of the trusted CES platform. In our prototype system, we have implemented the Ordering Buffer as a dedicated thread which buffers incoming trades in a priority queue (for ordering). When the OB has received all heartbeats up until a particular DC-derived timestamp it dequeues all the relevant trades to the Matching Engine over  shared-memory channels. 


\noindent\textbf{Scaling:} With higher numbers of MPs, a single OB instance would become the bottleneck (in aggregate, number of heartbeats scale linearly with participants). In such cases, scaling the OB is straightforward by leveraging sharding: multiple OB components could be deployed either as different threads on multicore CPUs or even as standalone VMs. Each OB needs to be responsible for a subset of the RBs. The OB instances can dequeue a batch of pending trades when safe and send them to ME-colocated OB for the final merge before they are forwarded to the matching engine. A distributed OB deployment would also allow handling the higher rates of heartbeats in the case of numerous MPs, as each OB can effectively filter out all incoming heartbeats before reaching the CES. Each distributed OB instance needs to maintain the minimum of current Delivery Clocks from its associated RBs, while the master OB needs to maintain the minimum DC from all the distributed OBs to be able to dequeue trades safely to the matching engine. Since contemporary cloud datacenter networks do not support in-network multicast for market data transmission, such distributed approach would also allow scaling the CES' market data distribution engine to higher rates.


\begin{figure}[t]
\centering
    \includegraphics[width=0.8\columnwidth]{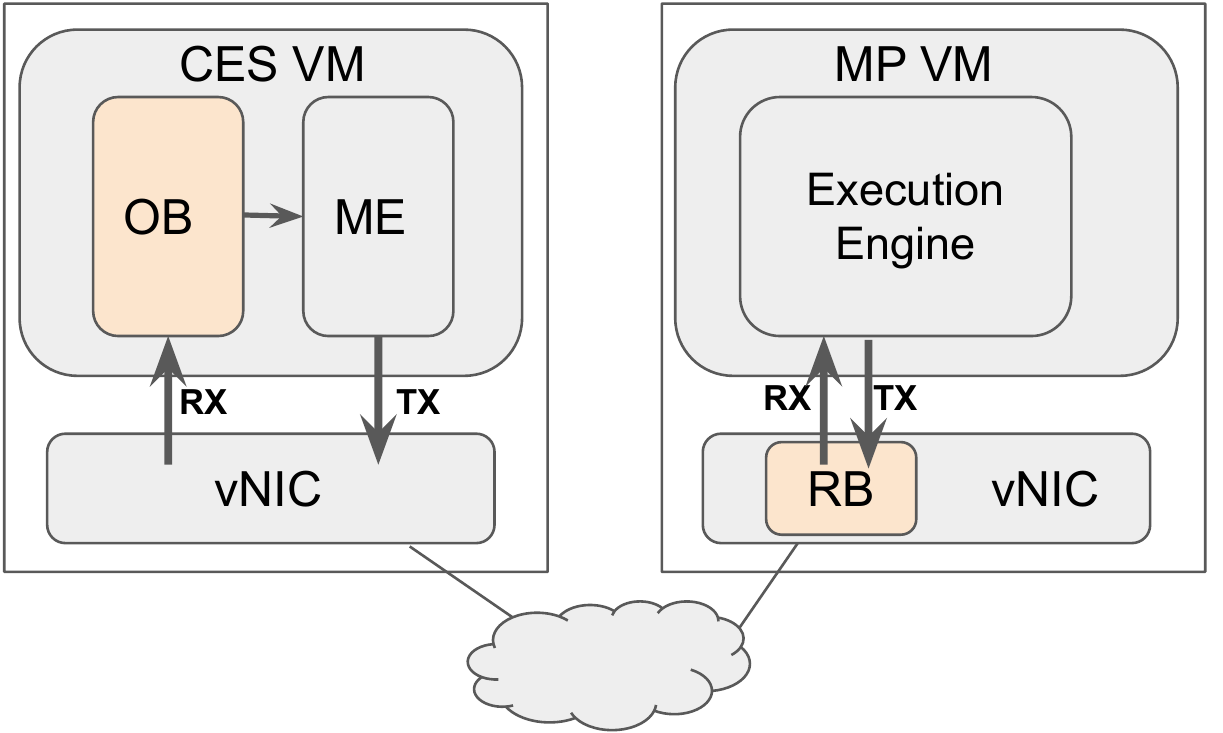}
    \caption{\small{\bf Cloud-hosted exchanges' architectural view.}}
    \label{fig:dbo_full_arch}
    \vspace{-5mm}
\end{figure}

\section{Evaluation}
\label{s:eval}

We evaluate the feasibility of our solution in hardware using our own hardware test bed. We use public-cloud experiments to get an understanding of overall DBO's performance in terms of latency and fairness if deployed.


\subsection{Methodology}
\label{ss:eval_methodology}

For all of the experiments (except simulation) presented in this section we leverage our prototype CES and MP implementations. On the CES side, we generate and distribute data to all Market Participants at fixed intervals. The market data arrive at the RBs, which later on release them to the  Market Participants. The MP implementation relies on busy-polling and kernel-bypass for low-latency access to the incoming market data packets, but does not utilize a sophisticated algorithm for trading decisions; it rather busy-waits for a pre-configured response time duration before generating a trade. We set each MP's  reaction time accordingly so that we can derive the expected final ordering at the OB and evaluate fairness. 

\noindent\textbf{Fairness metric:} For any number of MPs, perfect fairness is achieved when all competing trades among all unique pairs of participants are fully ordered (from faster to slower). We define the metric of fairness as the ratio of the number of competing trade sets  that were ordered correctly to the total number of competing trade sets for all unique pairs of market participants.


\noindent\textbf{End-to-end latency} We define end-to-end latency of a trade using Equation~\ref{eq:latency_def} ($F(i,a)-G(x)-RT(i,a)$). Generation time and forwarding time are measured at the CES. For the purpose of reporting latency  and fairness (and \emph{not} for ordering trades in DBO), we assume that the trigger point is known. We use it to calculate the response time of trades at the release buffer.


 We  evaluate our solution on three different setups: \begin{enumerate*}[label=(\alph*)]\item on-premise, bare-metal testbed deployment, \item public-cloud-based deployment, and \item simulation \end{enumerate*}.

Evaluation schemes: We evaluate three schemes. \begin{enumerate*}\item Direct delivery: This is the baseline scheme. There is no release buffer or ordering buffer and both trades and market data points just incur the underlying network latency. 
\item DBO: Based on our discussion in \S\ref{ss:understanding_latency}, we use $\delta=20$, $\kappa=0.25$ and $\tau=20 \mu s$.
\item CloudEx: CloudEx requires fine-grained clock synchronization, which is not available in our test-bed and cloud experiments. Due to inaccuracies in clock-synchronization, in our experiments we experience frequent release buffer and ordering buffer overruns. We only report results for CloudEx in simulation where we assume perfectly synchronized clocks. 
\end{enumerate*}.

\noindent
\textbf{Response Time:} 
The response time for each trade is a random number between 5 and 20 $\mu$ s and is within the horizon ($\delta$). 
Note that our solution does not ensure fairness for speed races where the response time (of the faster participant) is greater than the horizon. We picked a horizon to accommodate majority of the speed races. But we explicitly take into account this limitation. We present latency results with longer horizons and include experiments where the response time exceeds the horizon.

\subsection{Evalution on DPU-enabled baremetal servers}
\label{ss:on-premise-eval}

Our lab setup consists of three machines: one CES server and two MP servers. The CES server is equipped with an Nvidia ConnectX-5 NIC with two 100Gbps ports. Each MP server hosts one Nvidia BlueField-2 DPU with two 100Gbps ports. The server has a dual-CPU Intel Xeon processor running at 3.1 GHz. Each BlueField-2 DPU has eight ARMv8 A72 cores. All machines are connected via a 100GbE switch. We run Linux kernel (v5.4.0) and DPDK (v21.11) for the CES, RB, MP network engines.

The CES is generating market data every $40\mu s$ ($25K$ ticks per second), and the market participant servers are generating responses within $\delta$ time horizon since the reception of the data. The RB  is executing on the BlueField-2 DPU's SoC.

Table~\ref{tbl:bluefield} shows the achieved fairness and latency of our system. Direct delivery achieve poor fairness because of differences in network latency. DBO achieves perfect fairness at the cost of latency. In particular, to achieve response time fairness, the OB waits for the slowest participant. The latency is bounded by maximum network round-trip latency. This optimal latency bound (Theorem~\ref{thm:latency}) is shown as Max-RTT in the table. The difference between the optimal and DBO is becuase of batching, pacing and heartbeats.

\begin{table}[t]
\small
    \centering
    \begin{tabular}{c|c|cccc}
    & \textbf{Fairness} & \multicolumn{4}{c}{\textbf{Latency $(\mu s)$}} \\
    & \textbf{(\%)} & \textbf{avg} & \textbf{p50} & \textbf{p99} & \textbf{p999} \\
    \hline
       \textbf{Direct} & 74.62 & 9.60 & 9.52 & 16.58 & 25.25 \\
        \textbf{Max-RTT} &  - & 10.23 & 9.94 & 18.08 & 26.18 \\
       \textbf{DBO}  & 100 &  15.92 & 12.16 & 28.82 & 46.80 \\

    \end{tabular}
    \caption{\small{Fairness and trade latency results on bare metal servers with BlueField-based RB implementation.}} 
    \label{tbl:bluefield}
    \vspace{-9mm}
\end{table}

\subsection{Cloud-hosted Testbed}
\label{ss:cloud-eval}

We wish to understand how our system performs in a real public cloud-based deployment with several market participants. As discussed in \S \ref{ss:release_buffer}, we do not have access to the cloud providers' programmable NICs to deploy the RB functionality. To work around this limitation, we have adjusted our RB implementation so that it runs as a co-located  process with the market participant's execution engine on the MP VMs. In such configuration, the RB is using a kernel-bypass network stack to take over a dedicated vNIC which it uses to receive the UDP stream of market data from the CES, and to send back any trades submitted by the MP. To facilitate fast MP-to-RB communication we rely on shared-memory-based IPC primitives. Clearly, such a solution does not provide any security guarantees as the RBs run on VMs owned by the market participants which are not part of the CES' Trusted Computing Base, and could easily tamper with the RB's market data delivery engine or the delivery clock  measurements. It allows us, however, to evaluate the real-world performance (i.e., achievable throughput and latency) of our DBO system in a public cloud deployment.

\begin{table}[h!]
\small
    \centering
    \begin{tabular}{c|c|cccc}
    & \textbf{Fairness} & \multicolumn{4}{c}{\textbf{Latency $(\mu s)$}} \\
    & \textbf{(\%)} & \textbf{avg} & \textbf{p50} & \textbf{p99} & \textbf{p999} \\
    \hline
       \textbf{Direct} & 57.61 & 27.9 & 27.48 & 32.5 & 44.03 \\
        \textbf{Max-RTT}  & - & 33.34 & 32.44 & 42.01 & 48.38  \\
       \textbf{DBO}  & 100 &  47.19 & 46.95 & 55.71 & 67.41 \\

    \end{tabular}
    \caption{\small{Fairness and end-to-end latency for different schemes; full traces collected over a 15-minute duration. For consistency, Max-RTT is calculated using the packet timestamps from the DBO experiment trace.}}
    \label{tab:fairness_latency}
    \vspace{-7mm}
\end{table}

We set out to evaluate the fairness and end-to-end latency of different schemes. We deploy ten market participants and one CES as virtual machines (Standard\_F8s) in Microsoft Azure. We configure the  aggregate service rate to $125,000$ transactions (trades) per second (market data generation interval is fixed to $40\mu s$). Table~\ref{tab:fairness_latency} summarizes the achieved fairness and end-to-end latency results for direct delivery and DBO. 



\noindent\textbf{Fairness:} Direct delivery achieves poor fairness in our experiments. Compared to our test-bed where there is no network traffic and the variability in latency across participants is lower, direct delivery performs worse in the cloud.
DBO always achieve perfect fairness.
We discuss fairness for slow responders in \ref{ss:fairness_gt_delta}.

\noindent\textbf{Latency:}  As expected, direct delivery achieves  the lowest latency, at the cost of fairness. 
On the other hand, DBO trades off latency to achieve perfect fairness, but \emph{it still achieves sub-100us p999 tail latency in the public cloud.}
This latency is well within the requirements of many major exchanges. IEX, for example, a major exchange that prides itself on fairness had 700$\mu$ s latency~\cite{iex_cost_report}. We believe that with additional optimizations such as network traffic prioritization, in-network multicast, proximity placement groups, this number could be further brought down.The p9999 latency is much higher (\textasciitilde3.5ms);  full trace analysis shows that packet drop rate is very low but we identified a well-aligned, periodic queue buildup at the OB which we believe is due to scheduling artifacts in the VM. 



\subsubsection{Understanding DBO latency:}

\begin{figure}[t]
\centering
    \includegraphics[width=0.8\columnwidth]{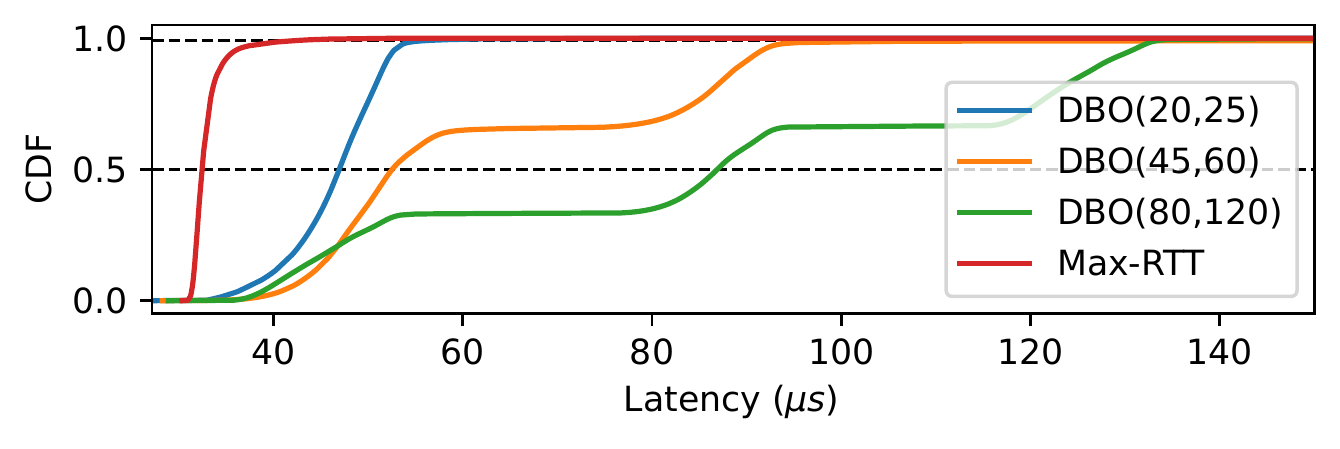}
    \vspace{-3mm}
    \caption{\small{\bf CDFs of the end-to-end trade  latency for various DBO configurations.}}
    \vspace{-3mm}
    \label{fig:latency_cdfs}
\end{figure}

How do DBO parameters affect end-to-end trade latency? Figure~\ref{fig:latency_cdfs} illustrates the CDF of the latency with different DBO configurations. Here, DBO(x,y) refers to using a horizon $\delta = x$ and batch size $(1+\kappa) \cdot \delta =y$. We also include the latency bound. As expected increasing the horizon and the batch size increases the latency. When batch size is 60 $\mu$s we see one inflection point. For batch size of 120 $\mu$s we see two inflection points. These inflection points are a direct result of batching. Since market data generation rate is 40, for batch size of 60, roughly $2/3$ of the batches contain  two data points. The first point in such batches incurs 40$\mu s$ of additional delay compared to second point. This difference create the inflection point.  Similarly for batch size 120, on average there are three market data points, the first point in the batch incurs an additional delay of 80 $\mu s$ while the second point incurs an additional delay of 40 $\mu s$. For batch size of 20, which contains only one market data point, the batching delay is zero. The deviation from the optimal latency bound is primarily due to hearbeats. Recall when network latency is well behaved, pacing does not add additional delay. Since $\tau$ is 20, the difference on average is 10$\mu s$

\subsubsection{Latency overhead with multiple participants\\}

DBO \textit{post} corrects for latency variance among participants at the OB. As the number of market participants increases, network RTT variance and asynchrony might be pronounced and can cause higher delays at the Ordering Buffer. In Figure~\ref{fig:ob_delay} we plot the time spent at the OB as a function of the number of MPs. Note, that the RB components remain unaffected as the number of participants gets higher. 

\begin{figure}[t]
\centering
    \includegraphics[width=0.8\columnwidth]{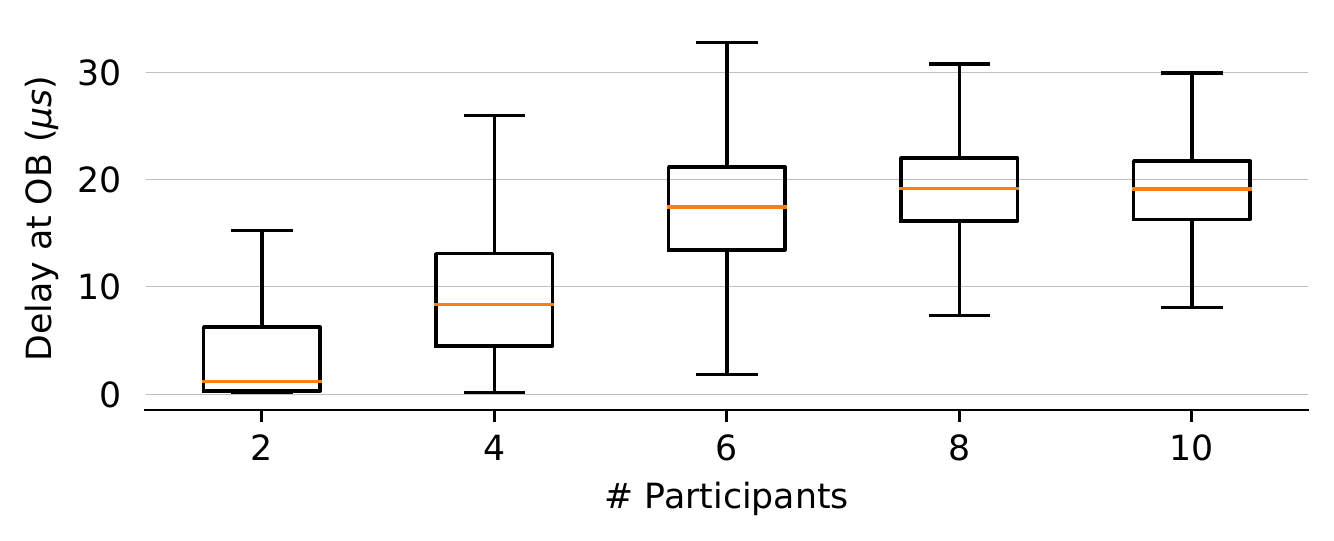}
    \vspace{-3mm}
    \caption{\small{\bf Trades' buffering delay at the Ordering Buffer.}}
    
    \label{fig:ob_delay}
    \vspace{-3mm}
\end{figure}

\subsubsection{Trades with response time > $\delta$\\}
\label{ss:fairness_gt_delta}
DBO only guarantees fairness for trades with a limited response time. Table~\ref{tab:fairness_rt} shows the fairness for such trades for different values of response time. In each experiment, the response time for the trade is derived from a range of values (shown on the top of the table). Direct delivery achieves poor fairness (similar to Table~\ref{tab:fairness_latency}). In contrast, even though the response time of trades exceeds the horizon $\delta$, DBO achieves close to ideal fairness. DBO orders such trades fairly, if the inter-delivery time for the batch that triggered the trade and the last batch corresponding to the trade is same across all participants. In the cloud experiments, even though latency differs across participants, for any particular participant the latency exhibits little variation (except a few latency spikes). As a result, the inter-delivery time for batches is similar (= $(1+\kappa) \cdot \delta$) across all participants for most of the time. DBO is thus able to correct for static differences in latency across participants and achieve fairness.

\begin{table}[h!]
\small
    \centering
    \begin{tabular}{c|cccccc}
    \textbf{RT (in $\mu s$)}  & 10-15 & 15-20 & 20-25 & 25-30 & 30-35 & 35-40\\
    \hline
       \textbf{Direct}  & 0.45 & 0.46 & 0.46 & 0.46 & 0.46 & 0.46\\

       \textbf{DBO} & 1.0 & 1.0 & 0.999 & 0.999 & 0.997 & 0.985 \\
    \end{tabular}
    \caption{\small{Fairness for trades with response time$> \delta = 20$.}}
    \label{tab:fairness_rt}
    \vspace{-8mm}
\end{table}

\begin{figure}[t]
  \centering
  \begin{subfigure}[b]{0.48\linewidth}
    \includegraphics[trim={0 0 0 1mm},clip,width=\linewidth]{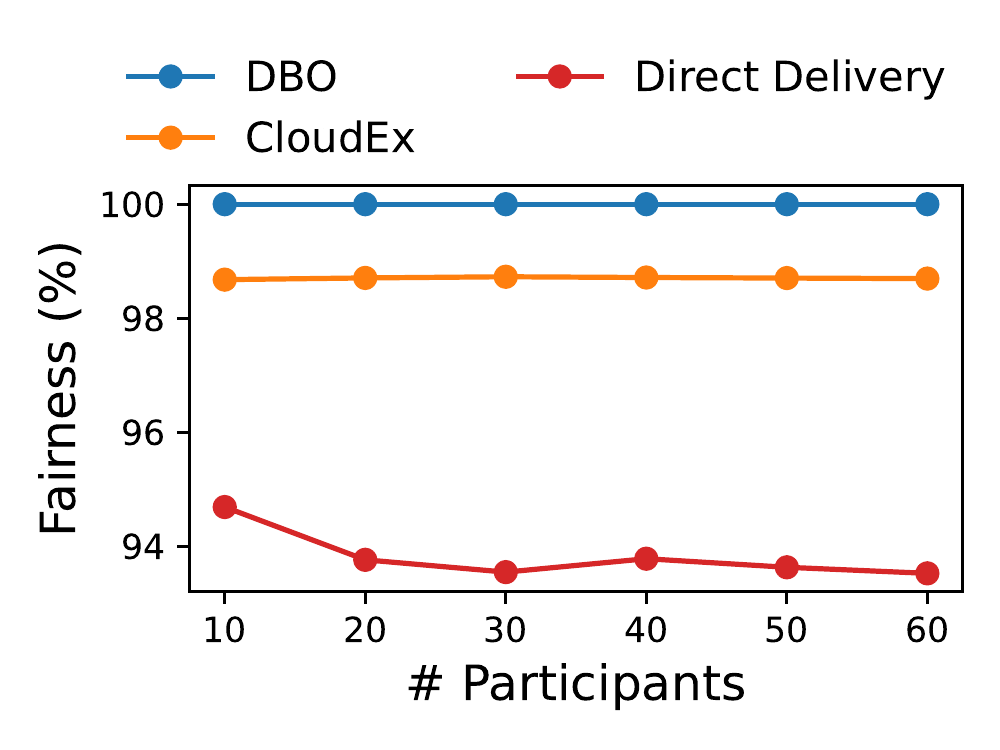}
    \vspace{-5.5mm}
    \caption{\small{Fairness}}
    \label{subfig:fairness}
  \end{subfigure}
  \begin{subfigure}[b]{0.48\linewidth}
    \includegraphics[trim={0 0 0 1mm},clip,width=\linewidth]{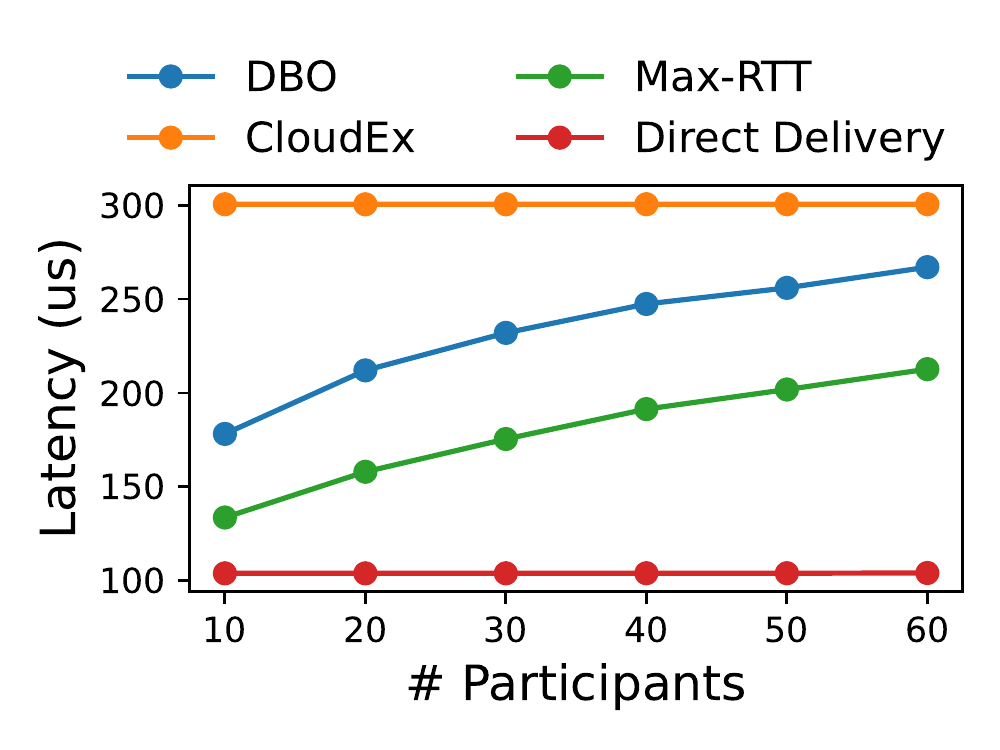}
    \vspace{-5.5mm}
    \caption{\small{Mean Latency}}
    \label{subfig:mean_latency}
  \end{subfigure}
  \vspace{-4mm}
  \caption{\small{\textbf{Comparison with CloudEx.}}}
  \vspace{-5mm}
  \label{fig:sim_participant}
\end{figure}

\subsection{Simulation}
\label{ss:simul}

We use simulation to compare DBO and CloudEx with perfect clock synchronization. We generated network traces to simulate latency between from the CES to RBs and from RBs to the OB. The minimum latency is $50 \mu s$ with random spikes in latency of up to 400 $\mu s$. Figure~\ref{fig:sim_trace} (in the Appendix) illustrates one such trace. For CloudEx we use 150 $\mu s$ threshold for delivery of market data and trade ordering. Note that the exact fairness and latency numbers are depend on the nature of network latency and the numbers here are illustrative.

\noindent
\textbf{Fairness with number of participants}
Compared to DBO, the latency for a trade in CloudEx is governed solely the round-trip latency of the participant. In contrast, DBO's latency is fundamentally bounded by maximum latency across participants. 

Figure~\ref{fig:sim_participant} shows the fairness and the average and latency for the two schemes as we scale the number of participants. The response times are between 5 to 20 $\mu s$. We also include direct delivery and max-RTT (latency bound) in the figure. CloudEx outperforms direct delivery however it incurs unfairness when latency spikes beyond the delivery threshold. As expected the fairness numbers for all the schemes are unaffected by the number of MPs. CloudEx incurs an average latency of 300 $\mu s$ of latency (sum of thresholds on the forward and the reverse path). The actual number is slightly higher due to latency spikes. The average latency does not scale with participants. In contrast, both max-RTT and DBO latency scale with participants. We find DBO achieves better latency on average than CloudEx. This is because when network latency is well behaved and close to the minimum (100 $\mu s$ round trip), the max-RTT is low and consequently the DBO latency is low. In contrast, in such scenarios CloudEx still incurs a minimum latency of 300 $\mu s$ in such cases.  However, in this experiment we find that tail latency of DBO and the latency bound (max-RTT) can both exceed the CloudEx latency (see Figure~\ref{fig:sim_tail_delta} in the Appendix). 
As explained earlier this is a fundamental cost to achieve fairness. 

\section{Conclusion}

In this paper we have presented DBO, a practical system to ensure fairness in the cloud for speed traders. DBO is incrementally deployable and achieves guaranteed fairness and low latency while still operating at high transaction rates. 

\bibliographystyle{abbrv} 
\begin{small}
\bibliography{hotnets22}
\end{small}

\clearpage
\appendix
\begin{sloppypar}
\section{Proof of Lemma~\ref{lemma:inter_delivery_imp}}
\label{app:lem1}

The lemma states that for response time fairness the inter-delivery times should be the same across all  MPs. 

\begin{proof}
To prove that the lemma condition is necessary we will show that if this condition is not met then no ordering system exists which can achieve response time fairness for arbitrary trade orders. 

\begin{figure}[t]
\centering
    \includegraphics[trim={0 0 0 1mm},clip,width=0.8\columnwidth]{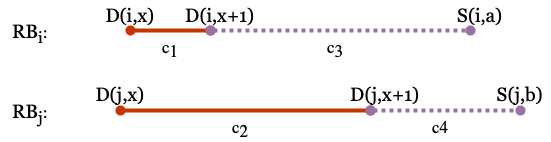}
    \vspace{-3mm}
    \caption{{\small{\bf Proof of Lemma ~\ref{lemma:inter_delivery_imp}.}} }
    \label{fig:proof}
    \vspace{-3mm}
\end{figure}

Consider the following scenario (Figure~\ref{fig:proof}) where the lemma condition is not met. Let $D(i,x+1) - D(i,x) = c1$, $D(j,x+1) - D(j,x) = c2$. Without loss of generality we assume $c1<c2$. 

Consider hypothetical trades $(i,a)$ and $(j,b)$ s.t. $S(i,a) = D(i, x+1) + c3$ and $S(j, b) = D(j, x+1) + c4$. Further, we can pick $S(i, a)$ and $S(j, b)$ s.t. $c3>c4$ and $c1+c3<c2+c4$. 
Now we consider two scenarios for how these trades were generated. These two scenarios are indistinguishable from the cloud provider/exchanges perspective.

\noindent
\text{Case 1:} $TP(i,a) = TP(j, b) = x+1$. Here,
\begin{align}
S(i,a)- D(i,x+1) = c3, S(j,b)-D(j,x+1) = c4.
\end{align}
Since $c3>c4$, condition $C1$  implies that,
$O(i,a) > O(j,b)$.

\noindent
\text{Case 2:} $TP(i,a) = TP(j, b)= x$. Here,
\begin{align}
S(i,a)- D(i,x) = c1+c3, S(j,b)-D(j,x) = c2+c4.
\end{align}
In this case, since $c1+c3<c2+c4$, for response time fairness the ordering must instead satisfy the opposite, $O(i,a) < O(j,b)$. A contradiction! \attn{Thus, no ordering system can achieve response time fairness in both these scenarios.}
\end{proof}

\section{Proof of Corollary~\ref{cor:inter_delivery_lrtf}}
\label{app:cor_inter_delivery_lrtf}

\begin{proof}
The proof is identical to that of Lemma~\ref{lemma:inter_delivery_imp}. The only difference being, we consider trades (i,a) , (j,b) and trigger point x and x+1, s.t., both $c1+c3$ and $c2+c4$ are less than $\delta$.
\end{proof}

\section{Proof of Theorem~\ref{thm:rb_to_mp_latency}}
\label{app:rb_to_mp_latency}

\begin{proof}
To the prove this theorem we will show that with DBO the ordering of trades $(i,a)$ and $(j,b)$ that meet the Theorem condition is $O(i,a) < O(j,b)$.

Consider a trade $(i,a)$ with response time less than $\delta - B_h(i)$. Let $\hat{D}(i,a)$ represent the delivery time at the RB. The observed submission time at RB ($\hat{S}(i,a)$) for such a trade will be,
\begin{align}
    \hat{S}(i,a) = \hat{D}(i,x) + RT(i,a) + RB\_MP\_L(i,x,a).
\end{align}
where $RB\_MP\_L(i,x,a)$ represents the combined network round trip latency between RB$_i$ and $MP_i$ for trigger point $x$ and trade $(i,a)$ . Because $RB\_MP\_L(i,x,a)$ is bounded by $B_h(i)$, $RT(i,a) + RB\_MP\_L(i,x,a) < \delta$ or $\hat{S}(i,a) < \hat{D}(i,x) + \delta$.

Recall, consecutive batches are atleast separated by $\delta$. This means that the trigger point ($x=TP(i,a)$) must be within the last received batch. The point $ld(i,a)$ is also the last point in this batch and $\hat{D}(i,ld(i,a)) = \hat{D}(i,x)$. The delivery clock for trade $(i,a)$ will thus be: $O(i,a) = DC(i,a) = \langle ld(i,a), RT(i,a) + RB\_MP\_L(i,x,a)\rangle$.

With batching, for participant $j$, $x$ and $ld(i,a)$ also belong to the same batch $\hat{D}(j,ld(i,a)) = \hat{D}(j,x)$.
For a competing trade $(j,b)$ with higher response time, the delivery clock at the time of submission will either read $O(j,b) = DC(j,b)) = \langle ld(i,a)), RT(j,b) + RB\_MP\_L(j,x,b)\rangle$ (if $(j,b)$ was submitted before the next batch, i.e., $\hat{S}(j,b) < \hat{D}(j,ld(i,a)+1)$) or $DC(j, b) = \langle y, \hat{S}(j,b)-\hat{D}(j,y)\rangle$ with $y>ld(i,a)$.

C3 implies that, $RT(i,a) < RT(j,b) - (B_h(i)- B_l(j))$ and $ B_l(i) \leq RB\_MP\_L(i,x,a) \leq B_h(i), B_l(j) \leq RB\_MP\_L(j,x,b) \leq B_h(j)$. As a result, $ RT(i,a) + RB\_MP\_L(i,x,a) < RT(j,b) + RB\_MP\_L(j,x,b)$

As a result, in both the cases, $O(i,a) < O(j,b)$. Hence proved.
\end{proof}

\section{Impact of Losses}
\label{app:impact_losses}
\noindent
\textbf{Impact of market data points being lost:} Like status-quo we advocate market participants requesting any dropped market data points separately. The retransmitted market data point does not update the delivery clock at the release buffer. This way trades generated using the retransmitted data points get affected. However, fairness for all other trades remains unaffected. 
The latency of the system can get affected as the delivery clock of the participant experiencing losses lags transiently until the next data point is delivered. If data points are generated infrequently, then the delivery clock of the participant might take a large time to recover. To prevent this explicitly, we advocate CES sending periodic heartbeats. However, we believe that major exchanges already generate data at a very high frequency (a data point every 20  $\mu s$) and such heartbeats are not necessary.


\noindent
\textbf{Impact of trades being lost:} In the event a trade is lost, the participant can retransmit the trade. The retransmitted trade will be tagged by the delivery clock at the time of the retransmission. Such a retransmitted trade will incur unfairness. However, fairness of all other trades remains unaffected.

\noindent
\textbf{Impact of heartbeats being lost:} Lost hearbeats do not impact fairness. However, if a heartbeat is lost then the OB might have to wait an additional time (for the next heartbeat to arrive) before forwarding the trades to the CES increasing latency (Equation~\ref{eq:latency_def}).

\section{Thwarting front-running attacks}
\label{app:front_running}





%
%

We impose two simple constraints on communication to preven front running. \begin{enumerate*}[label=(\arabic*)]\item A participant machine and its helper machines can communicate with each other freely but they cannot communicate with any other machines in the cloud. This restriction can be imposed easily by cloud providers today using security groups. This restriction ensures that a participant machine cannot get market data from other participant machines in the cloud directly. Next, we will ensure that a participant machine cannot get an earlier market data feed from outside the cloud. 
We will do so by restricting that a participant can only send data point x out of the cloud, when x has been delivered to all participants in the cloud. This way, market data points can only be available outside the cloud once they have been delivered to all the participants.
\item The helper machines cannot send data outside the cloud. Any data (excluding the trade orders) from a participant being sent outside the cloud is tagged by the delivery clock at the RB and buffered at a gateway. The data sent by the participant could potentially be a market data point with id less than or equal to the last point id (first tuple) of the delivery clock time stamp. The gateway thus buffers this data until it is sure that the all data points with id less than the last data point id in the delivery clock time stamp have been delivered. For this purpose, RB's periodically communicate their delivery clock to the gateway. 
%
\end{enumerate*}


\section{Latency for network traces used in simulation}

\begin{figure}[t]
\centering
    \includegraphics[width=0.95\columnwidth]{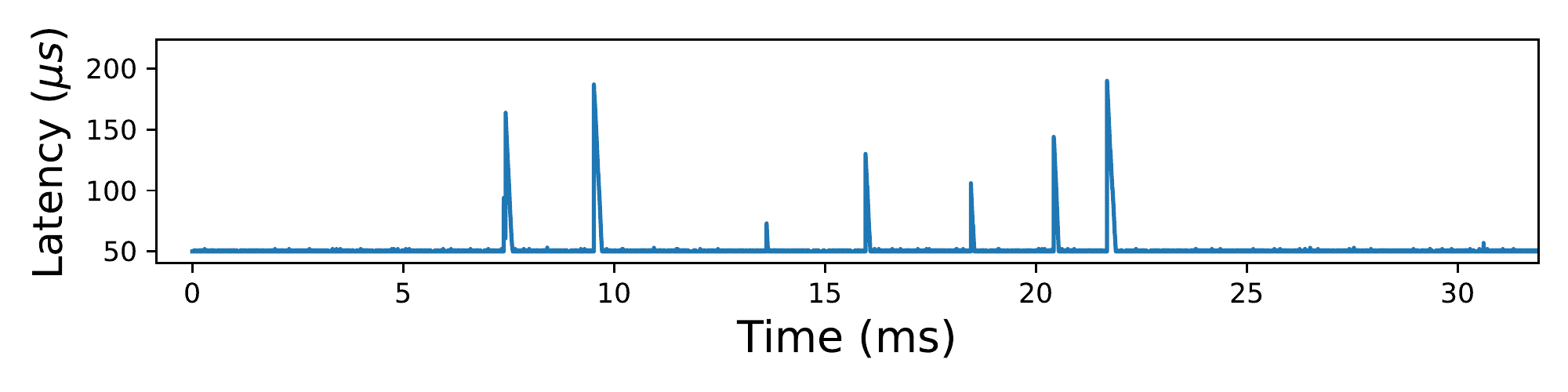}
    \caption{\small{\bf Network trace used in simulation.}}
    \label{fig:sim_trace}
\end{figure}

\section{Tail Latency with number of participants}

\begin{figure}[t]
\centering
    \includegraphics[width=0.8\columnwidth]{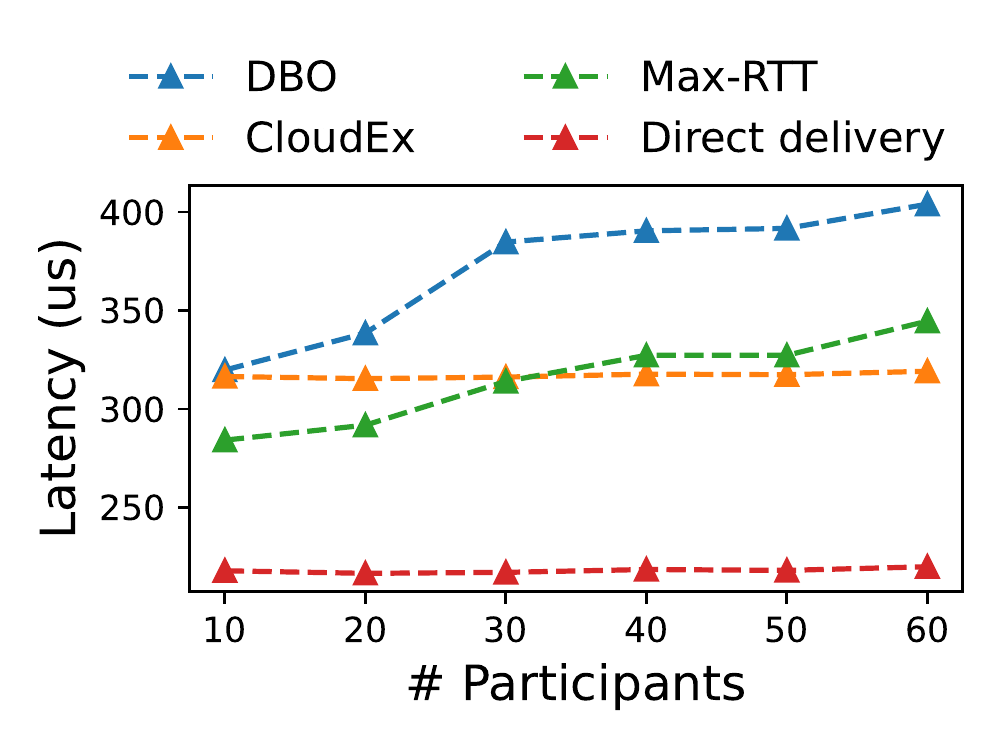}
    \vspace{-4mm}
    \caption{\small{\bf Fairness and Latency trends with number of participants.}}
    \label{fig:sim_tail_delta}
\end{figure}
\end{sloppypar}

\end{document}